\documentclass[12pt,a4paper]{article}
\usepackage{setspace}
\usepackage[english]{babel}
\usepackage{amsfonts,amsmath,amssymb,graphicx,xcolor,here,bm,hyperref,accents,color,float,dsfont}
\usepackage{natbib} \bibpunct{(}{)}{;}{a}{,}{,} 
\usepackage{tikz}
\usepackage{enumerate}
\usepackage{caption} 
\usepackage{subcaption} 
\newtheorem{theorem}{Theorem}

\newtheorem{corollary}[theorem]{Corollary}

\newtheorem{definition}[theorem]{Definition}
\newtheorem{example}[theorem]{Example}

\newtheorem{lemma}[theorem]{Lemma}

\newtheorem{remark}[theorem]{Remark}

\newenvironment{proof}[1][Proof]
{\textbf{#1.} }{\ \rule{0.5em}{0.5em}}

\def\qed{\hfill \vrule height 6pt width 6pt depth 0pt\\}

\allowdisplaybreaks[3]
\newcommand{\dif}{\mathrm{d}}
\def\EE{\mathbb{E}}
\def\PP{\mathbb{P}}
\def\p{\mathbb{P}}

\renewcommand{\(}{\left(}
\renewcommand{\)}{\right)}
\renewcommand{\[}{\left[}

\allowdisplaybreaks

\def\qed{\hfill \vrule height 6pt width 6pt depth 0pt}

\newcommand{\id}{\mathds{1}}

\newcommand{\vcovar}{\mathrm{VCoVaR}}
\newcommand{\vcoes}{\mathrm{VCoES}}

\newcommand{\mcoes}{\mathrm{MCoES}}

\def\R{\mathbb{R}}

\DeclareMathAccent{\wtilde}{\mathord}{largesymbols}{"65}
\hypersetup{
    colorlinks,
    linkcolor={red!50!black},
    citecolor={blue!50!black},
    urlcolor={blue!80!black},
    pdfborder={0 0 0}
}

\numberwithin{theorem}{section}

\usepackage{footnote}
\usepackage[affil-it]{authblk}
\usepackage[english]{babel}
\usepackage{blindtext}

\begin{document}

\title{\vskip -1cm \textbf{On Vulnerability Conditional Risk Measures: Comparisons and Applications in Cryptocurrency Market}}

\author[a]{Tong Pu}
\author[b]{Yunran Wei}
\author[a]{Yiying Zhang}%

\affil[a]{Department of Mathematics, Southern University of Science and Technology, Shenzhen, Guangdong, China}
\affil[b]{School of Mathematics and Statistics, Carleton University, Ottawa, Canada}
\date{\today}
\maketitle

\begin{abstract}
    We introduce a novel class of systemic risk measures, the Vulnerability Conditional risk measures, which try to capture the ``tail risk" of a risky position in scenarios where one or more market participants is experiencing financial distress. Various theoretical properties of Vulnerability Conditional risk measures, along with a series of related contribution measures, have been considered in this paper. We further introduce the backtesting procedures of VCoES and MCoES. Through numerical examples, we validate our theoretical insights and further apply our newly proposed risk measures to the empirical analysis of cryptocurrencies, demonstrating their practical relevance and utility in capturing systemic risk. 

    \noindent
    \\[1mm]
\end{abstract}
\baselineskip17.5pt
\thispagestyle{empty}


\newpage
\section{Introduction}
In the span of the last twenty years, a series of evolving scenarios and critical incidents, notably the financial crisis in 2007-2009, have illustrated the pronounced volatility, fragile, and interconnected nature of the system. These phenomena have led to the emergence of systemic risks, which are risks that could be suffered by financial institutions interconnected with the domino effect as a result of a systemic risk chain reaction. In order to minimize the emergence of financial risks as much as possible, Basel III, which is the current international regulatory framework, requires financial institutions like banks to hold capital that meets Value-at-Risk (VaR) and Expected Shortfall (ES) capital requirements to ensure they have enough capital to absorb losses without posing risks to the financial system. The measures of VaR and ES, by their very design, are incapable of capturing the ripple effects of some institutions' financial distress on other institutions or the system at large. As a result, while individual risks may be properly dealt with in normal times, the system itself still remains, or in some cases is induced to become, fragile and vulnerable to large macroeconomic shocks. Consequently, there is a great interest in the development of alternative risk measures that can effectively address this limitation and provide a more comprehensive view of the systemic risks.
\par In the list of measures proposed for the quantification of systemic risk, the Conditional Value-at-Risk (CoVaR), which is introduced by \cite{Adrian2016}, has emerged as the most widely utilized market-based measure. It is defined as the VaR of one specific financial institution, conditional upon the occurrence of an event corresponding to a stress scenario to which another financial institution is exposed. The measure of CoVaR is proficient at assessing the impact of the failure of one asset on another, provided that practitioners have chosen the appropriate condition. However, in practice, financial practitioners often need to measure the impact of a range of institutions or the entire market on a single asset, a situation where CoVaR may fall short, as it does not capture the aggregate influence of multiple assets or the whole market's effect on the asset in question. Recent studies have expanded the CoVaR framework by incorporating more than one variable in the conditional event, in order to take into account the impact of a series of financial institutions or the entire market on a single institution. \cite{cao2013multi} introduces the Multi-CoVaR (MCoVaR) with the condition of several institutions being simultaneously in distress. \cite{bernardi2021allocation} propose the System-CoVaR (SCoVaR), in which the conditional variables are aggregated via their sum. Vulnerability-CoVaR (VCoVaR), which considers the VaR of an institution with the condition that one of the selected institutions being in distress, is introduced by \cite{waltz2022vulnerability}. In addition, \cite{waltz2022vulnerability} provides a comprehensive empirical study and analyses how different distressing events of the cryptocurrencies impact the risk level of each other.  Interested readers can refer to \cite{mba2024assessing}, \cite{pu2024joint} and \cite{Xu2017SystemicRI} for more related systemic risk measures and investigations on the applications in insurance and finance.
\par Although the aforementioned article has made significant contributions to the definition and application of novel systemic risk measures, we must highlight two areas where it falls short. These systemic risk measures primarily focus on the VaR measures of systemic risks, neglecting the equally important and widely used ES measure, which may lead to an underestimation of risk to some extent. Meanwhile, the papers focus mainly on market applications without exploring the mathematical properties of these measures. Ignoring the mathematical properties results in a lack of deeper insight into these risk measures.
\par The main goal of this paper is to propose a formalization of conditional risk measures that take into account the influence of multiple financial institutions or the entire market on a specified financial institution and to introduce related contribution measures to consider the absolute and relative spillover effects from a market to a specified institution. We establish sufficient conditions for comparing these measures for two sets of random vectors with both different marginal distributions and different copulas. Our main results demonstrate that for a financial institution, having a higher level of risk or a closer connection to the market will increase the absolute and relative spillover effects of market risk on it. This is in line with our intuition. We provide an application study for the proposed risk measures and the aforementioned ones. As an important by-product of this research, backtesting methodologies for the newly developed conditional risk measures are provided in this study. 
\par The paper is structured as follows: Section \ref{sec:Preliminaries} provides the preliminaries for the conclusion of the article, Section \ref{sec:definition} formally presents the definition and mathematical expressions of the vulnerability conditional risk measures to be studied in this paper, Section \ref{sec:results} investigates the mathematical properties of the risk measures and correspondingly provides numerical examples, Section \ref{sec:backtesting} presents the backtesting methodologies for these measures, and Section \ref{sec:application} provides an application of these measures in the cryptocurrency market.


\section{Preliminaries} \label{sec:Preliminaries}
Throughout this paper, the terms ``increasing'' and ``decreasing'' are used in a mild sense. Expectations and integrals are assumed to exist whenever they appear. Let $(\Omega,\mathcal{F})$ be a measurable space and let $\mathcal{X}$ be the space of bounded random variables. For $F\in \mathcal{X}$, the generalized inverse is defined as 
$$F^{-1}(t) = \inf\{x\in \R : F(x)\ge t\},~ t\in (0,1].$$
\subsection{Copula and dependence notions}
Let $(X_1, X_2, \cdots, X_d)$ be a random vector with joint cumulative distribution function (c.d.f.) $F$, joint survival function (s.f.) $\overline{F}$, and respective absolutely continuous marginal c.d.f.s' $F_i$, $i\in \{1, 2, \cdots, d\}$. The joint c.d.f. $F$ can be expressed as
\begin{equation} \label{eq:copula-def}
    F\left(x_1, \ldots, x_d\right)=C\left(F_1\left(x_1\right), \ldots, F_d\left(x_d\right)\right), \forall (x_1, x_2, \cdots, x_d) \in \mathbb{R}^d,
\end{equation}
where $C$ is the unique copula of $(X, Y)$, that is, the joint c.d.f. of $(U, V)$, where $U = F(X)$ and $V = G(Y)$. 
\par The Archimedean copula, a prevalent category within the family of copulas, is characterized by a generating function known as the Archimedean generator. The expression for an $d$-dimensional Archimedean copula is given by: 
\begin{equation*}
    C_\psi (u_1, u_2, \ldots, u_d) = \psi^{-1} \left( \psi(u_1) + \psi(u_2) + \cdots + \psi(u_d) \right),
\end{equation*}
where $\psi$ is a strictly decreasing function called the generating function, with its inverse denoted as $\psi^{-1}$.
The Gumbel copula, which is a prominent example of Archimedean copulas and has the generator $\psi(t) = (-\log(t))^{\theta}, \theta \in [1,+\infty)$, is defined by  
\begin{equation} \label{eq:gumbel}
    C(u_1,u_2,u_3;\theta)=\exp\(-\((-\log(u_1))^\theta + (-\log(u_2))^\theta +(-\log(u_3))^\theta \)^{1/\theta}\),
\end{equation}
for $u_i \in [0,1]$ where $i \in \{1,2,3\}$.
\begin{definition} \label{def:LTD}
(\cite{Nelsen2007}) Let $\bm{X}=\left(X_1, X_2, \cdots, X_n\right)$ be a $d$-dimensional random vector, and let the sets $A$ and $B$ partition $\{1,2, \cdots, d\}$.
\begin{enumerate}[(i)]
    \item $\operatorname{SI}\left(\mathbf{X}_B \mid \mathbf{X}_A\right)$ if $P\left[\mathbf{X}_B>\mathbf{x}_B \mid \mathbf{X}_A=\mathbf{x}_A\right]$ is nondecreasing in $\mathbf{x}_A$ for all $\mathbf{x}_B$
    \item When $\operatorname{SI}\left(\mathbf{X}_B \mid \mathbf{X}_A\right)$ holds for all singleton sets $A$, i.e., $A=\{i\}, i=1,2, \cdots, n$; then $\mathbf{X}$ is positive dependent through the stochastic ordering (PDS); 
    \item $\operatorname{LTD}\left(\bm{X}_B \mid \bm{X}_A\right)$ if $\PP\left\{\bm{X}_B \leq \bm{x}_B \mid \bm{X}_A \leq \bm{x}_A\right\}$ is nonincreasing in $\bm{x}_A$ for all $\bm{x}_B$.
    \item $\operatorname{LTD}^{1}_{m}$ if $\PP\left\{{X}_i \leq {x}_i, \forall i \in \{1,2, \cdots, n\} ~\text{and}~ i \neq m \mid {X}_m \leq {x}_m \right\}$ is nonincreasing in ${x}_m$ for all $x_1,x_2, \cdots, x_{m-1}, x_{m+1}, \cdots, x_n$.
\end{enumerate}
\end{definition}
\par The notation $\operatorname{LTD}^{1}_{m}$ signifies that as the value of $X_m$ decreases, the probability of the other variables being in the lower orthant increases, which implies that the probability of these variables taking on smaller values is enhanced. As a result, $\operatorname{LTD}$ characterizes the positive dependence structure of a random vector. 

\begin{lemma} \label{lem:ltd}
    \begin{enumerate}[(i)]
        \item A $d$-dimensional copula $C$ is of $\operatorname{LTD}^{1}_{m}$ if $C(u_1,\cdots,u_d) / u_m$ nonincreasing in $u_m$ for all $u_1, u_2, \cdots , u_{m-1}, u_{m+1}, \cdots , u_d$.
        \item A $d$-dimensional Archimedean copula $C_\psi$ is of $\operatorname{LTD}^{1}_{m}$ for all $1 \leq m \leq d$ if $x \psi ' (x)$ is none-decreasing.
    \end{enumerate}
   
\end{lemma}
\begin{lemma} \label{lem:pds}
    Let $\bm{X}=\left(X_1, X_2, \cdots, X_d\right)$ be a $d$-dimensional random vector that admits an absolutely continuous copula $C$. Then $\bm{X}$ is PDS if and only if $C$ is componentwise concave.
\end{lemma}
\subsection{Stochastic orders and distortion risk measures}
\par Recall that a distortion function $h : [0, 1] \to [0, 1]$ is an non-decreasing function such that $h(0) = 0$ and $h(1) = 1$. The set of all distortion functions is henceforth denoted by $\mathcal{H}$. The class of distortion risk measures is then defined as follows. 
\begin{definition}\label{def:dist-risk-measure} 
For a distortion function $h \in \mathcal{H}$ and a r.v. $X$ with c.d.f. $F$, the distortion risk measure $\mathrm{D}_h$ is defined as
\begin{equation} \nonumber
\mathrm{D}_h[X]=\int_0^{+\infty} h(\overline{F}(t)) \dif t-\int_{-\infty}^0[1-h(\overline{F}(t))] \dif t.
\end{equation}
\end{definition}

\par Given distortion function $h \in \mathcal{H}$, we call $\hat{h}(x)=1-h(1-t)$  the dual distortion function.
\par Stochastic orders are partial orders defined on sets of c.d.f.’s and serve as a powerful tool for comparing different  random variables. Interested readers can refer to \cite{shaked2007sto} and \cite{belzunce2015introduction} for more details.
\begin{definition} \label{def:order}
     \citep{shaked2007sto,Denuit2005} Let $X$ and $Y$ be two random variables (r.v.'s) with respective c.d.f.'s $F$ and $G$, survival functions (s.f.'s) $\overline{F}$ and $\overline{G}$, and probability density functions (p.d.f.'s) $f$ and $g$, respectively. $X$ is said to be smaller than $Y$ in the
\begin{enumerate}[(i)]
\item usual stochastic order (denoted by $X \leq_{\rm st} Y$) if $\overline{F}(x) \leq \overline{G}(x)$ for all $x \in \mathbb{R}$, or equivalently, ${\rm D}_h[X] \leq {\rm D}_h[Y]$ for all distortion functions $h$;
\item increasing convex order (denoted by $X \leq_{\rm icx} Y$) if $\mathbb{E}[\phi(X)] \leq \mathbb{E}[\phi(Y)]$ for any increasing and convex function $\phi: \mathbb{R} \rightarrow \mathbb{R}$, or equivalently, ${\rm D}_h[X] \leq {\rm D}_h[Y]$ for all concave distortion functions $h$;
\item dispersive order (denoted by $X \leq_{\rm disp} Y$) if $\mbox{\rm VaR}_v[X]-\mbox{\rm VaR}_u[X] \leq \mbox{\rm VaR}_v[Y]-\mbox{\rm VaR}_u[Y] $, for all $0<u \leq v<1$;
\item star order (denoted by $X \leq_{\star} Y$) if $G^{-1}(u)/F^{-1}(u)$ is increasing in $u \in (0,1)$.
\item expected proportional shortfall (in short \emph{EPS}) order (denoted by $X \leq_{\rm eps} Y$) if
\begin{equation}\nonumber
  \mbox{\rm EPS}_p[X]=\mathbb{E}\left[\left(\frac{X-\mbox{\rm VaR}_p[X]}{\mbox{\rm VaR}_p[X]}\right)_+\right]\leq \mathbb{E}\left[\left(\frac{Y-\mbox{\rm VaR}_p[Y]}{\mbox{\rm VaR}_p[Y]}\right)_+\right]=\mbox{\rm EPS}_p[Y],
\end{equation}
for all $p\in\{q\in(0,1):\mbox{\rm VaR}_q[X]\neq0~\mbox{\rm and}~\mbox{\rm VaR}_q[Y]\neq0\}$.
\end{enumerate}
\end{definition}

\section{VCoVaR and VCoES and their associated contribution measures} \label{sec:definition}

In this section, we review and introduce some vulnerability conditional risk measures and their associated contribution counterparts. Before formally introducing the conditional measures, the univariate VaR and the CoVaR measures are reviewed. Let $X$ and $Y$ be risks faced by financial institutions. The VaR at probability level $\beta \in (0,1)$ is implicitly defined by 
\begin{equation} \nonumber
    \PP( Y \leq {\rm VaR}_{\beta}(Y)) = \beta. 
\end{equation}
If $Y$ have c.d.f. $F$, one can alternatively write ${\rm VaR}_{\beta}(X) = F^{-1}(\beta)$. The CoVaR is defined by modify VaR by adding inequality to the condition:
\begin{equation} \nonumber
    {\rm CoVaR}_{\alpha,\beta}(Y \vert X) = {\rm VaR}_{\beta} (Y \vert X > {\rm VaR}_{\alpha}(X)).
\end{equation}
\subsection{VCoVaR and associated contribution measures}
%

\par  

Consider a set of $d+1$ risks faced by multiple financial institutions denoted by $(X_1,\ldots,X_d,Y)$, for $d \in \mathbb{N}$. 
\begin{definition}\label{def:vulner} 
Given the stress events that at least one of $\bm{X}=(X_1,\ldots,X_d)$ exceeds their stress levels regulated by {\rm VaR}'s with confidence levels $\alpha_1,\ldots,\alpha_d$, the vulnerability conditional value-at-risk (in short {rm VCoVaR}): $\mathcal{X} \rightarrow \R$ is given by
\begin{equation}\label{VCoVaR-defi}
{\rm VCoVaR}_{\bm{\alpha},\beta}(Y|\bm{X})={\rm VaR}_{\beta}(Y|\mbox{$\exists$ $i$: $X_i>{\rm VaR}_{\alpha_i}(X_i)$}),
\end{equation}
where $\alpha_i\in [0,1)$ and $\beta\in(0,1)$.
\end{definition}
\par The VCoVaR measure defined in (\ref{VCoVaR-defi}) is firstly introduced in \cite{waltz2022vulnerability} to investigate the properties of and the systemic risks in the cryptocurrency market. They investigated important theoretical findings of this measure, and implemented some empirical studies to show how different distressing events of the cryptocurrencies impact the risk level of each other. Note that the stress event is characterized by a broader event that at least one financial institution is at bankruptcy level and thus it presents all possible distress scenarios and is hence maybe more appropriate in capturing domino effects  (contagion effects induced by systemic risks) than some existing alternatives such as the MCoVaR (cf. \cite{ortega2021stochastic}).

Assume that $(\bm{X},Y)$ has a $d+1$-dimensional copula $C$. Let $U_1,\ldots,U_d,V$ be a group of uniform random variables such that $(\bm{U},V)$ also shares copula $C$, where $\bm{U}=(U_1,\ldots,U_d)$. We first present an alternative expression of (\ref{VCoVaR-defi}). For ease of presentation, we denote $$A_{\bm{X}}=\{\mbox{$\exists$ $i$: $X_i>\mbox{VaR}_{\alpha_i}(X_i)$}\},\quad A^{>}_{\bm{X}}=\{\mbox{$\forall$ $i$: $X_i>\mbox{VaR}_{\alpha_i}(X_i)$}\}$$ 
and
$$A_{\bm{U}}=\{\mbox{$\exists$ $i$: $U_i>\alpha_i$}\},\quad A^{>}_{\bm{U}}=\{\mbox{$\forall$ $i$: $U_i>\alpha_i$}\}.$$
Clearly, $A^{>}_{\bm{X}}$ is a subset of $A_{\bm{X}}$ and $A_{\bm{X}}$ means that at least one financial institution is at distress, while $A^{>}_{\bm{X}}$ means all of the $d$ financial institutions are at distress. Then, the following result can be reached. 
\par For given $d+1$-dimensional copula $C$, we denote $C(\alpha_1,\alpha_2,\cdots,\alpha_d,t)$ as $C(\bm{\alpha},t)$ where $\bm{\alpha} = (\alpha_1,\alpha_2,\cdots,\alpha_d)$ for simplicity. 

\begin{lemma}\label{lem:cdf-representation}
\begin{enumerate}[(a)]
\item The distribution function of $V|A_{\bm{U}}$ is given by 
\begin{equation}\label{eq:F-V-AU}
F_{V|A_{\bm{U}}}(v)=\frac{v-C(\bm{\alpha},v)}{1-C(\bm{\alpha},1)},\quad v\in[0,1].
\end{equation}
Further, (\ref{eq:F-V-AU}) is a distortion function.
\item The distribution function of $Y|A_{\bm{X}}$ is  
\begin{equation}\label{eq:F-Y-AX}
F_{Y|A_{\bm{X}}}(y)=F_{V|A_{\bm{U}}}(F_{Y}(y)).
\end{equation}
\end{enumerate}
\end{lemma}
\par The following Corollary provides the mathematical expression for VCoVaR.
\begin{corollary} \label{coro:VMCoVaR-representation}
    Suppose that both $(\bm{X},Y)$ and $(\bm{U},V)$ have copula $C$, where $U_1,\ldots,U_d,V$ are standard uniform random variables. Then, $\mbox{\rm VCoVaR}_{\bm{\alpha},\beta}(Y|\bm{X})$ and admit the following expressions
    \begin{equation} \label{eq:vcovar-representation}
        {\rm VCoVaR}_{\bm{\alpha},\beta}(Y|\bm{X})=F^{-1}_{Y}\left(F^{-1}_{V|A_{\bm{U}}}(\beta)\right).
    \end{equation}
\end{corollary}
\par Risk spillover effects refer to the phenomenon where risks from one financial institution, asset, or market segment spread to others, potentially causing a broader impact on the financial system. This can occur due to various factors such as interconnectedness through financial instruments, common exposures to certain market conditions, or the influence of one institution's failure on the confidence and stability of others. Relative risk spillover and absolute risk spillover are two ways to consider the risk spillover effects. According to Definition \ref{def:vulner}, we can define the related contribution measures as follows to characterize the relative and absolute risk spillover effects.
\begin{definition}\label{def:vulnercontri1} The vulnerability contribution conditional value-at-risk of risk $Y$ given the event $A_{\bm{X}}$ compared with the conventional value-at-risk of $Y$, written as $\Delta{\rm VCoVaR}$, is defined as follows:
\begin{equation}\label{contriVCoVaR-defi1}
\Delta{\rm VCoVaR}_{\bm{\alpha},\beta}(Y|\bm{X})={\rm VCoVaR}_{\bm{\alpha},\beta}(Y|\bm{X})-{\rm VaR}_{\beta}(Y).
\end{equation}
Moreover, the associated vulnerability contribution ratio conditional value-at-risk is defined as follows:
\begin{equation}\label{contriratioVCoVaR-defi1}
\Delta^R{\rm VCoVaR}_{\bm{\alpha},\beta}(Y|\bm{X})=\frac{{\rm VCoVaR}_{\bm{\alpha},\beta}(Y|\bm{X})-{\rm VaR}_{\beta}(Y)}{{\rm VaR}_{\beta}(Y)}.
\end{equation}
\end{definition}

According to \eqref{eq:vcovar-representation}, mathematical expressions of (\ref{contriVCoVaR-defi1}) and (\ref{contriratioVCoVaR-defi1}) are given as follows:
\begin{equation}\label{ConVCoVaR-expre}
\Delta{\rm VCoVaR}_{\bm{\alpha},\beta}(Y|\bm{X})=F^{-1}_{Y}\left(F^{-1}_{V|A_{\bm{U}}}(\beta)\right)-F^{-1}_{Y}(\beta).
\end{equation}
and 
\begin{equation}\label{ConratioVCoVaR-expre}
\Delta{\rm VCoVaR}_{\bm{\alpha},\beta}(Y|\bm{X})=\frac{F^{-1}_{Y}\left(F^{-1}_{V|A_{\bm{U}}}(\beta)\right)-F^{-1}_{Y}(\beta)}{F^{-1}_{Y}(\beta)}.
\end{equation}

\subsection{VCoES and associated contribution measures}
To consider a more complete picture of the potential losses in the tail of the distribution under extreme conditions, we define the so-called vulnerability conditional expected shortfall (written as VCoES) risk measures as follows, as a direct generalization of the VCoVaR \eqref{VCoVaR-defi} proposed by \cite{waltz2022vulnerability}. 
\begin{definition}\label{def:vulnerCOES} The vulnerability conditional expected shortfall of risk $Y$ given some stress event induced by $\bm{X}=(X_1,\ldots,X_d)$ is defined as follows:
\begin{equation}\label{VCoES-defi}
{\rm VCoES}_{\bm{\alpha},\beta}(Y|\bm{X})=\frac{1}{1-\beta}\int_{
\beta}^1{\rm VCoVaR}_{\bm{\alpha},t}(Y|\bm{X})\dif t.
\end{equation}
\end{definition}
In a similar manner, it can be verified that
\begin{equation}\label{VCoES-defi-Equi}
    \begin{split}
{\rm VCoES}_{\bm{\alpha},\beta}(Y|\bm{X})&=\frac{1}{1-\beta}\int_{\beta}^1F^{-1}_{Y}\left(F^{-1}_{V|A_{\bm{U}}}(t)\right)\dif t\\
&=\int_{0}^1F^{-1}_{Y}\left(F^{-1}_{V|A_{\bm{U}}}(t)\right)\dif \bar{h}_{TVaR}(t)\\
&= \int_{0}^1F^{-1}_{Y}(p)\dif \bar{h}_{TVaR}(F_{V|A_{\bm{U}}}(p)),
    \end{split}
\end{equation}
where $F_{V|A_{\bm{U}}}(\cdot)$ is defined in (\ref{eq:F-V-AU}). 
See \cite{sordo2018IME} for detailed studies on CoES, which can be considered as a special case of VCoES by setting $d=1$. 



\begin{definition}\label{def:vulnercontri3} 
    The vulnerability contribution conditional expected shortfall of risk $Y$ given the event $A_{\bm{X}}$, written as $\Delta{\rm VCoES}$, is defined as follows:
        \begin{equation}\label{contriVCoES-defi1}
        \Delta{\rm VCoES}_{\bm{\alpha},\beta}(Y|\bm{X})={\rm VCoES}_{\bm{\alpha},\beta}(Y|\bm{X})-{\rm ES}_{\beta}(Y).
        \end{equation}
        Moreover, the associated vulnerability contribution ratio conditional value-at-risk is defined as follows:
        \begin{equation}\label{contriratioVCoES-defi1}
        \Delta^R{\rm VCoES}_{\bm{\alpha},\beta}(Y|\bm{X})=\frac{{\rm VCoES}_{\bm{\alpha},\beta}(Y|\bm{X})-{\rm ES}_{\beta}(Y)}{{\rm ES}_{\beta}(Y)}.
        \end{equation}
        \end{definition}
        
        According to \eqref{eq:vcovar-representation}, alternative expressions of (\ref{contriVCoES-defi1}) and (\ref{contriratioVCoES-defi1}) are given as follows:
    \begin{equation}\nonumber
        \Delta{\rm VCoES}_{\bm{\alpha},\beta}(Y|\bm{X})=\frac{1}{1-\beta}\int_{\beta}^1 F^{-1}_{Y}\left(F^{-1}_{V|A_{\bm{U}}}(t)\right) - F^{-1}_{Y}(t) \dif t,
    \end{equation}
    and 
    \begin{equation}\nonumber
        \Delta^{\rm R}{\rm VCoES}_{\bm{\alpha},\beta}(Y|\bm{X})= \frac{\int_{0}^1 F^{-1}_{Y}\left(p\right)\dif \bar{h}_{TVaR}(F_{V|A_{\bm{U}}}(p))}{\int_{\beta}^1 F^{-1}_{Y}(p) \dif \bar{h}_{TVaR}(p)} -1.
    \end{equation}

\section{Stochastic orders and vulnerability conditional risk measures} \label{sec:results}
\subsection{Main results}
In this section, we present the mathematical properties of VCoVaR and VCoES, as well as their associated contribution measures. To begin with, with a fixed vector $\bm{\alpha}$, it can be readily observed that both VCoVaR and VCoES are increasing with respect to $\beta$. The monotonicity of VCoVaR can be derived by applying Corollary \ref{coro:VMCoVaR-representation} and considering that the functions defined in \eqref{eq:F-V-AU} and \eqref{eq:F-V-AUg} are non-decreasing w.r.t. $v$. And \eqref{VCoES-defi-Equi} give rise to the monotonicity of VCoES with respect to $\beta$. 
\par The following results address this question: Under what circumstances will the spillover effect of risk be greater? Our common sense suggests that when financial institutions hold greater risks or when individual financial institutions are more closely connected to the market, the spillover effect of risk will be greater. The results below support the common intuition. 
\par In the following discussion, we consider two $d+1$-demensional random vectors $(\bm{X}_1,Y_1)$ and $(\bm{X}_2,Y_2)$ with respective copulas $C_1$ and $C_2$. Let $U^*_1,\ldots,U^*_d,V^*$ and $\hat{U}_1,\ldots,\hat{U}_d,\hat{V}$ be two groups of uniform random variables such that $(\bm{U}^*,V^*)$ and $(\hat{\bm{U}},\hat{V})$ respectively admit copula $C_1$ and $C_2$, where $\bm{U}^*=(U^*_1,\ldots,U^*_d)$ and $\hat{\bm{U}}=(\hat{U}_1,\ldots,\hat{U}_d)$. For given $\bm{\alpha} \in [0,1]^d$, let 
\begin{equation} \label{eq:def_l_alpha}
    l_{\bm{\alpha}}(v)=\frac{v-C_1(\bm{\alpha},v)}{v-C_2(\bm{\alpha},v)}.
\end{equation}
Next, we investigate the monotonicity property of ${\rm VCoVaR}_{\bm{\alpha},\beta}(Y|\bm{X})$ with respect to the confidence levels, dependence structure and marginal distributions of risks.
\begin{theorem} \label{thm:vcovar-dcdm}
    If $Y_1 \leq_{{\rm st}} Y_2$ and 
    \begin{equation} \label{eq:l-alpha}
        l_{\bm{\alpha}}(v) \geq l_{\bm{\alpha}}(1), \quad \forall v \in [0,1],
    \end{equation}
    then 
    \begin{equation} \label{eq:VCoVaR-leq}
        {\rm VCoVaR}_{\bm{\alpha},\beta}(Y_1|\bm{X}_1) \leq {\rm VCoVaR}_{\bm{\alpha},\beta}(Y_2|\bm{X}_2)
    \end{equation}
    for all $\beta \in (0,1)$.
\end{theorem}
\par It is worthy noting that the requirement \eqref{eq:l-alpha} in Theorem \ref{thm:vcovar-dcdm} is quite general and will be automatically satisfied if $l_{\alpha}(t)$ increases w.r.t. $t \in [\beta,1]$. In some parametric settings, this conditions is satisfied when $C_2$ is of greater positive dependency, comparing with $C_1$. For example, assume that $C_1$ is Gumbel copula \eqref{eq:gumbel} with parameter $\theta_1$ and $C_2$ is Gumbel copula with parameter $\theta_2$, it can be readily verified that \eqref{eq:l-alpha} is satisfied when $\theta_1 \leq \theta_2$.

\par We next provide sufficient conditions for comparing the VCoVaR-associated contribution measures for two sets of random vectors $(\bm{X}_1,Y_1)$ and $(\bm{X}_2,Y_2)$.
\begin{theorem} \label{thm:d-drvcovar-dcdm}
    Suppose one (or both)  of $C_1$ and $C_2$ is ${\rm LTD}_{d+1}^1$ and \eqref{eq:l-alpha} is satisfied. 
    \begin{enumerate}[(i)]
        \item If $Y_1 \leq_{{\rm disp}} Y_2$, then $\Delta{\rm VCoVaR}_{\bm{\alpha},\beta}(Y_1|\bm{X}_1) \leq \Delta{\rm VCoVaR}_{\bm{\alpha},\beta}(Y_2|\bm{X}_2)$ for all $\beta \in (0,1)$.
        \item If $Y_1 \leq_{{\star}} Y_2$, then $\Delta^{\rm R}{\rm VCoVaR}_{\bm{\alpha},\beta}(Y_1|\bm{X}_1) \leq \Delta^{\rm R}{\rm VCoVaR}_{\bm{\alpha},\beta}(Y_2|\bm{X}_2)$ for all $\beta \in (0,1)$.
    \end{enumerate}
\end{theorem}

The sufficient conditions for comparing the VCoES's and the related contribution measures for two sets of random vectors are provided in the next two theorems. Similar results for CoES can be found in \cite{mainik2014} and \cite{sordo2018IME} and treated as special case of (i) of Theorem \ref{thm:dr-vcoes-dcdm}.
\begin{theorem} \label{thm:dr-vcoes-dcdm}
    Suppose that one (or both) of $C_1$ and $C_2$ is concave w.r.t. its last argument and \eqref{eq:l-alpha} is satisfied.
    \begin{enumerate}[(i)]
        \item If $Y_1 \leq_{\rm icx} Y_2$, then ${\rm VCoES}_{\bm{\alpha},\beta}(Y_1|\bm{X}_1) \leq {\rm VCoES}_{\bm{\alpha},\beta}(Y_2|\bm{X}_2)$ for all $\beta \in (0,1)$.
        \item If $Y_1 \leq_{\rm eps} Y_2$, then $\Delta^{\rm R}{\rm VCoES}_{\bm{\alpha},\beta}(Y_1|\bm{X}_1) \leq \Delta^{\rm R}{\rm VCoES}_{\bm{\alpha},\beta}(Y_2|\bm{X}_2)$ for all $\beta \in (0,1)$.
    \end{enumerate}
\end{theorem}
\par The following result, serving as a corollary to Theorem \ref{thm:vcovar-dcdm} and Theorem \ref{thm:dr-vcoes-dcdm}, reveals to us an intuitive result: under positive dependence, VCoVaR and VCoES are always greater than the unconditional prototypes VaR and ES, respectively. As a fundamental property of VCoVaR and VCoES, this was not mentioned in \cite{waltz2022vulnerability} nor, to the best of our knowledge, other related papers, and we supplement it here.
\begin{corollary}\label{cor:VCoVaR-greater}
    If $C_1$ is of ${\rm LTD}_{d+1}^1$, then ${\rm VaR}_{\beta}(Y_1) \leq {\rm VCoVaR}_{\bm{\alpha},\beta}(Y_1|\bm{X}_1)$ and ${\rm ES}_{\beta}(Y_1) \leq {\rm VCoES}_{\bm{\alpha},\beta}(Y_1|\bm{X}_1)$.
\end{corollary}
\par We next provide sufficient conditions for comparing the VCoES contribution measures for two sets of random vectors.
\begin{theorem} \label{thm:dvcoes-dcdm}
    Suppose that $(\bm{X}_1,Y_1)$ and $(\bm{X}_2,Y_2)$ admit respective copulas $C_1$ and $C_2$ and one (or both) of $C_1$ and $C_2$ is ${\rm LTD}_{d+1}^1$. If $Y_1 \leq_{{\rm disp}} Y_2$ and \eqref{eq:l-alpha} is satisfied, then 
    \begin{equation} \label{eq:DVCoES-leq}
        \Delta{\rm VCoES}_{\bm{\alpha},\beta}(Y_1|\bm{X}_1) \leq \Delta{\rm VCoES}_{\bm{\alpha},\beta}(Y_2|\bm{X}_2)
    \end{equation}
for all $\beta \in (0,1)$.
\label{VCoES_ratio_beta}
\end{theorem}

%

\subsection{Numerical examples} \label{sec:Numerical}
\par In this subsection, we provide some numerical examples to illustrate our main findings. Given a random variable $X$, it is said that $X$ follows a Pareto distribution with shape parameter $a>0$ and scale parameter $k \in \mathbb{R}$, denoted by $X \sim {\rm Pareto}(a, k)$, if its survival function is given by
\begin{equation} \nonumber
    \bar{F}(x)=\left(\frac{k}{x}\right)^a, \quad \text { for all } x \in(k,+\infty).
\end{equation}
Suppose $Y_1 \sim {\rm Pareto}(a_1, k_1)$ and $Y_2 \sim {\rm Pareto}(a_2, k_2)$. According to Tables 2.1-2.2 in \cite{belzunce2015introduction}, the following statements hold:
\begin{enumerate}[(a)]
    \item $Y_1 \leq_{\rm st} Y_2$ when $a_1 \geq a_2$ and $ k_1 \leq k_2$;
  \item $Y_1 \leq_{\rm icx} Y_2$ when $a_1, a_2 > 1$, $k_1> (=) k_2$ and $\frac{a_1\left(a_2-1\right)}{a_2\left(a_1-1\right)}=(<) \frac{k_2}{k_1}$;
  \item $Y_1 \leq_{\rm disp} Y_2$ when $a_1 \geq a_2$ and $a_1 k_2 \geq a_2 k_1$; 
  \item $Y_1 \leq_{\rm \star} Y_2$ when $a_1 \geq a_2$; and hence $Y_1 \leq_{\rm epw} Y_2$ when $a_1 \geq a_2$.
\end{enumerate}
\par Below we present some examples illustrating Theorems \ref{thm:vcovar-dcdm} - \ref{VCoES_ratio_beta}. 
\begin{example}
Let $d=2$ and $\alpha_1 = \alpha_2 = 0.95$. 
Assume that $(\bm{X}_1, Y)$ admits the Gumbel Copula \eqref{eq:gumbel}
with $\theta = 2$ and $(\bm{X}_2, Y)$ admits the Gumbel Copula with $\theta = 3$. It can be verified by Lemma \ref{lem:pds} that both of the two copulas are PDS. In this paper, we introduce risk measures that are not contingent upon the distribution of variables in the conditions, thus enabling us to concentrate exclusively on setting the distributions of $Y_1$ and $Y_2$ without the need to posit any assumptions about the distributions of $\bm{X}_1$ and $\bm{X}_2$.
\begin{enumerate}[(i)]
    \item Set $Y_1 \sim {\rm Pareto}(20, 16)$ and $Y_2 \sim {\rm Pareto}(16, 20)$. Thus, it holds that $Y_1 \leq_{\rm st} Y_2$. Figure \ref{fig:sub1} shows that the $\vcovar_{\bm{\alpha}, \beta}(Y_1|\bm{X}_1) \leq \vcovar_{\bm{\alpha}, \beta}(Y_2|\bm{X}_2)$ for all $\beta \in (0,1)$, which validates the results of Theorem \ref{thm:vcovar-dcdm}. In addition, it is readily seen that the VCoVaR measures are greater than the respective unconditional prototypes. This fact illustrates the result of Corollary \ref{cor:VCoVaR-greater}.
    \item Set $Y_1 \sim {\rm Pareto}(4, 5)$ and $Y_2 \sim {\rm Pareto}(3, 4)$. It holds that $Y_1 \leq_{\rm disp} Y_2$ but $Y_1 \nleq_{\rm st} Y_2$. As illustrated in Figure \ref{fig:sub2}, ${\rm \Delta VCoVaR}_{\bm{\alpha}, \beta}(Y_1|\bm{X}_1) \leq {\rm \Delta VCoVaR}_{\bm{\alpha}, \beta}(Y_2|\bm{X}_2)$. Hence, the effectiveness of (i) of Theorem \ref{thm:d-drvcovar-dcdm} is validated.
    \item Set $Y_1 \sim {\rm Pareto}(4, 3)$ and $Y_2 \sim {\rm Pareto}(3, 2)$. It holds that $Y_1 \leq_{\star} Y_2$ but $Y_1 \nleq_{\rm st} Y_2$ nor $Y_1 \nleq_{\rm disp} Y_2$. As demonstrated in Figure \ref{fig:sub3}, ${\rm \Delta^R VCoVaR}_{\bm{\alpha}, \beta} (Y_1|\bm{X}_1) \leq {\rm \Delta^R VCoVaR}_{\bm{\alpha}, \beta} (Y_2|\bm{X}_2)$ is confirmed, thus affirming the potency of (ii) of Theorem \ref{thm:d-drvcovar-dcdm}.
    \item Set $Y_1 \sim {\rm Pareto}(9, 20)$ and $Y_2 \sim {\rm Pareto}(5, 18)$. It holds that $Y_1 \leq_{\rm icx} Y_2$ but $Y_1 \nleq_{\rm st} Y_2$ nor $Y_1 \nleq_{\rm disp} Y_2$. With $\vcoes_{\bm{\alpha}, \beta}(Y_1|\bm{X}_1) \leq \vcoes_{\bm{\alpha}, \beta}(Y_2|\bm{X}_2)$ established through Figure \ref{fig:sub4}, this validates the operational strength of (i) of Theorem \ref{thm:dr-vcoes-dcdm}. In addition, it is readily seen that the {\rm VCoES} measures are greater than the {\rm ES}. This fact illustrates the result of Corollary \ref{cor:VCoVaR-greater} as well.
    \item Set $Y_1 \sim {\rm Pareto}(4, 5)$ and $Y_2 \sim {\rm Pareto}(3, 4)$. It holds that $Y_1 \leq_{\rm disp} Y_2$ but $Y_1 \nleq_{\rm st} Y_2$. The validation of Theorem \ref{thm:dvcoes-dcdm}'s effectiveness is supported by the establishment of ${\rm \Delta VCoES}_{\bm{\alpha}, \beta} (Y_1|\bm{X}_1) \leq {\rm \Delta VCoES}_{\bm{\alpha}, \beta} (Y_2|\bm{X}_2)$ as shown in Figure \ref{fig:sub5}.
    \item Set $Y_1 \sim {\rm Pareto}(4, 5)$ and $Y_2 \sim {\rm Pareto}(3, 4)$. It holds that $Y_1 \leq_{\rm eps} Y_2$ but $Y_1 \nleq_{\rm st} Y_2$ nor $Y_1 \nleq_{\rm disp} Y_2$. The plot in Figure \ref{fig:sub6} supports that ${\rm \Delta^R VCoES}_{\bm{\alpha}, \beta} (Y_1|\bm{X}_1) \leq {\rm \Delta^R VCoES}_{\bm{\alpha}, \beta} (Y_2|\bm{X}_2)$, which validates the practicality of (ii) of Theorem \ref{thm:dr-vcoes-dcdm}.
\end{enumerate}
\end{example}

\begin{figure}[ht]
    \centering
    \begin{subfigure}[b]{0.3\textwidth}
        \includegraphics[width=\textwidth]{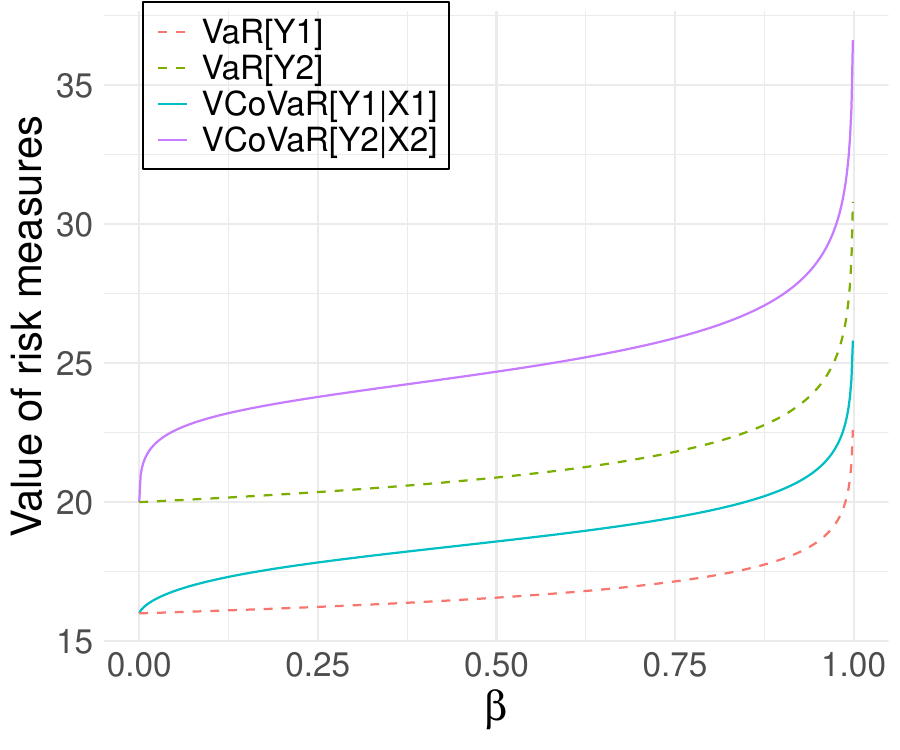}
        \caption{}
        \label{fig:sub1}
    \end{subfigure}
    \hfill 
    \begin{subfigure}[b]{0.3\textwidth}
        \includegraphics[width=\textwidth]{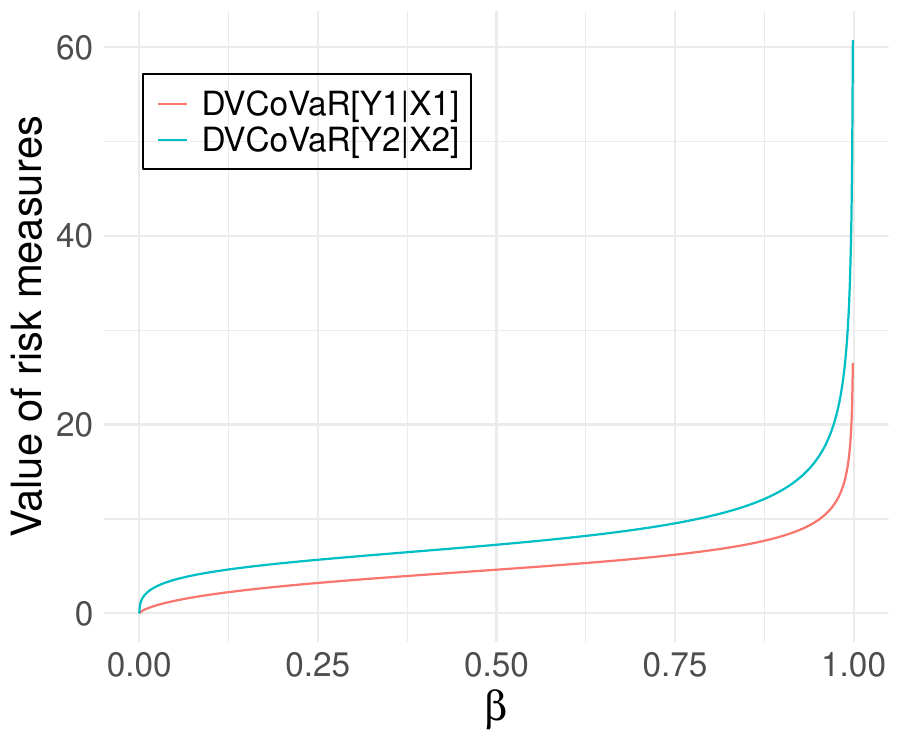}
        \caption{}
        \label{fig:sub2}
    \end{subfigure}
    \hfill 
    \begin{subfigure}[b]{0.3\textwidth}
        \includegraphics[width=\textwidth]{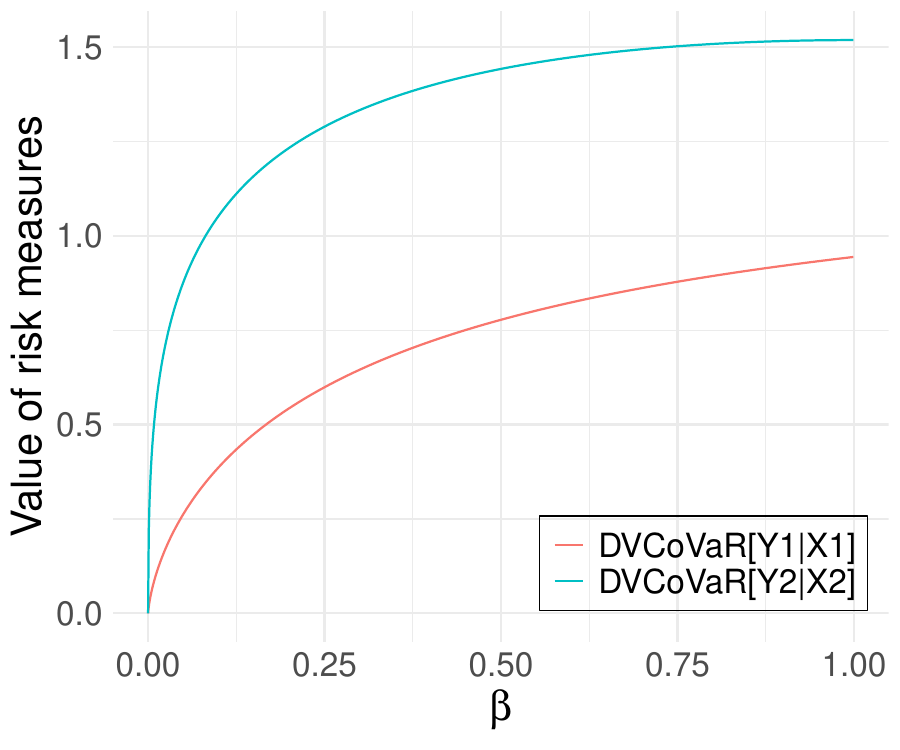}
        \caption{}
        \label{fig:sub3}
    \end{subfigure}
    \caption{Comparisons of $\vcovar_{\bm{\alpha}, \beta}(Y|\bm{X})$, ${\Delta VCoVaR}_{\bm{\alpha}, \beta}(Y|\bm{X})$ as functions of $\beta \in (0,1)$.}
    \label{fig:VCoVaR-comp}
\end{figure}

\begin{figure}[ht]
    \centering
    \begin{subfigure}[b]{0.3\textwidth}
        \includegraphics[width=\textwidth]{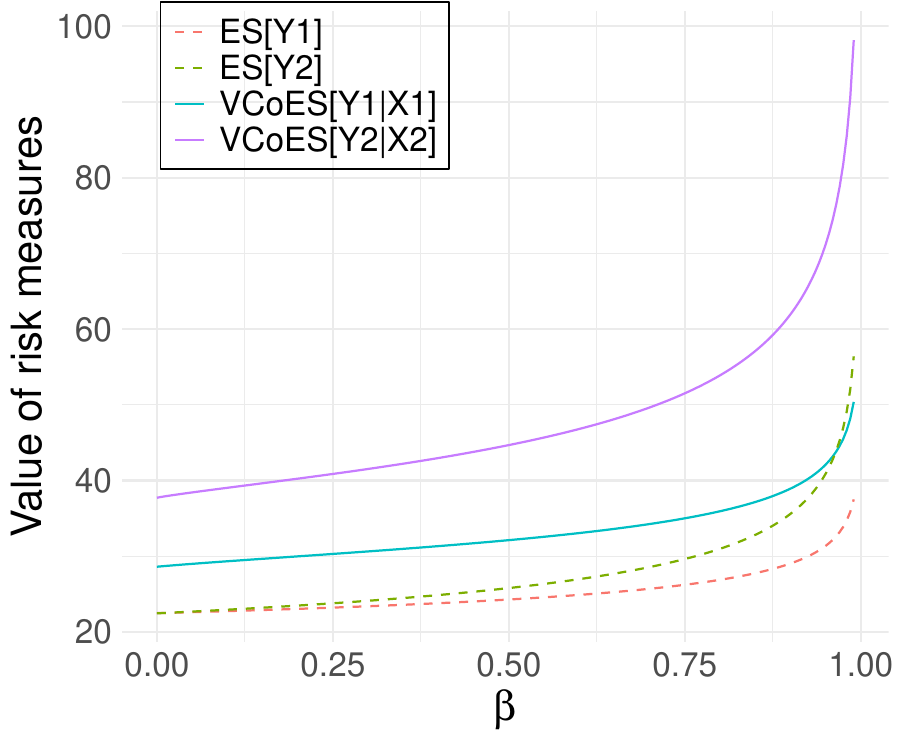}
        \caption{}
        \label{fig:sub4}
    \end{subfigure}
    \hfill 
    \begin{subfigure}[b]{0.3\textwidth}
        \includegraphics[width=\textwidth]{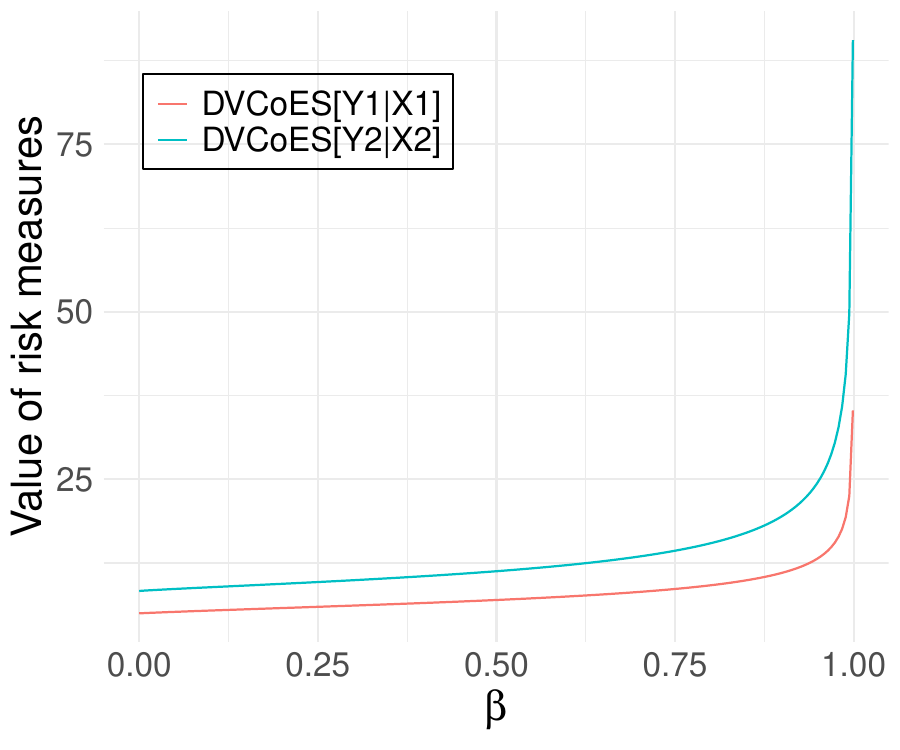}
        \caption{}
        \label{fig:sub5}
    \end{subfigure}
    \hfill 
    \begin{subfigure}[b]{0.3\textwidth}
        \includegraphics[width=\textwidth]{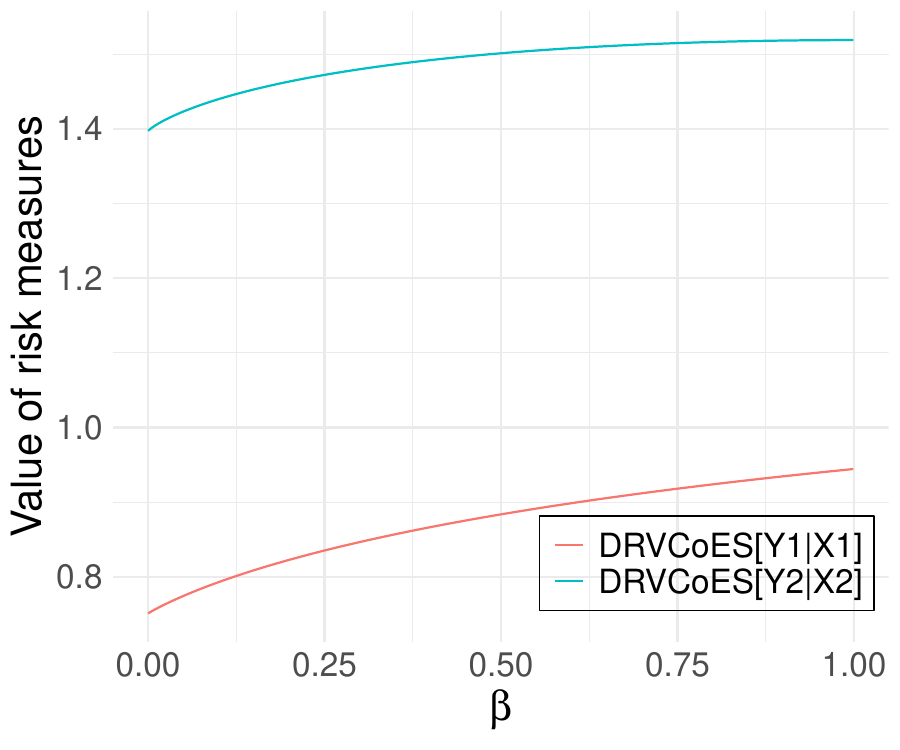}
        \caption{}
        \label{fig:sub6}
    \end{subfigure}
    \caption{Comparisons of $\vcovar_{\bm{\alpha}, \beta}(Y|\bm{X})$, ${\rm \Delta VCoVaR}_{\bm{\alpha}, \beta}(Y|\bm{X})$ as functions of $\beta \in (0,1)$.}
    \label{fig:VCoES-comp}
\end{figure}

\section{Backtesting $\vcoes$ and $\mcoes$} \label{sec:backtesting}
We introduce the backtesting procedures for $\vcoes$ and $\mcoes$ in this section. 
Below we list some notations and assumptions. Let $(\bm{X}_t, Y_t) = (X_{1,t},...,X_{d,t},Y_t)$ denote the set of risks at time $t$. 
Denote the set $A_{\bm{X}_t}$ by $A_{\bm{X}_t} = \{\mbox{$\exists$ $i$: $X_{i,t}>\mbox{VaR}_{\alpha_i}(X_{i,t})$}\}$. We use $\vcovar_{\bm{\alpha}, \beta, t}$ to denote the VCoVaR estimate by using the predictive model. 

Consider the violation-based test for VCoVaR and define the violation indicator variable by $\id_{t} = \id_{\{Y_{t}>\vcovar_{\bm{\alpha}, \beta, t}|A_{\bm{X}_{t}} \}}$, which is consistent with the backtesting procedures used in \cite{waltz2022vulnerability}. Following the multinomial backtest dicussed in \cite{kratz2018multinomial}, we backtest $\vcoes_{\bm{\alpha},\beta}$ by simultaneously backtesting $\vcovar_{\bm{\alpha}, \beta_1},...,\vcovar_{\bm{\alpha}, \beta_m}$, where $\beta_j = \beta + \frac{j-1}{m}(1-\beta),~j=1,...,m$.  We define the sequence of violation indicator variables: $\id_{t,j} = \id_{\{Y_{t}>\vcovar_{\bm{\alpha}, \beta_j, t}|A_{\bm{X}_{t}} \}},$ where $j=1,...,m$. We also set $\beta_0 = 0$ and $\beta_{m+1}=1$.

If the underlying predictive model is correct, then by \cite{christoffersen1998evaluating}, for fixed $j$, it holds that,
\begin{enumerate}[(i)]
    \item the unconditional coverage hypothesis, $\EE[\id_{t,j}] = 1-\beta_j$ for all $t$;
    \item the independence hypothesis, $\id_{t,j}$ is independent of $\id_{s,j}$ for $s\neq t$.
\end{enumerate}
We can proceed to define $$Z_t = \sum_{j=1}^m \id_{t,j},$$ where $Z_t$ counts the number of violations at time $t$. The sequence $Z_t$ satisfy the following two conditions:
\begin{enumerate}[(i)]
    \item the unconditional coverage hypothesis, $\p(Z_t\le i) = \beta_{i+1},~i=0,...,m$ for all $t$;
    \item the independence hypothesis, $Z_t$ is independent of $X_s$ for $s\neq t$.
\end{enumerate}

 Let $N$ denote the sample size, which is the set of sample satisfying the condition $A_{\bm{X}_t}$. We use $O_j$ to count the number of violations: $O_j = \sum_{t=1}^n \id_{\{Z_t =j\}},~j=0,1,...,m$. Then under the given conditions,   $(O_0,...,O_m)$ follows the multinomial distribution $\rm{MN}(N, \beta_1 -\beta_0,...,\beta_{m+1}-\beta_m)$. Consider the model $(O_0,...,O_m)$ follows the multinomial distribution $\rm{MN}(N, \theta_1 -\theta_0,...,\theta_{m+1}-\theta_m)$, where $0 = \theta_0 <\theta_1<...<\theta_m <\theta_{m+1}=1$. 
 We test the null and alternative hypotheses given by
 \begin{align*}
\left\{
\begin{array}{l}
H_0: \theta_j = \beta_j \textrm{ for } j=1,...,m \\
H_1: \theta_j \neq \beta_j \textrm{ for at least one } j\in \{1,...,m\}.  \\
\end{array}
\right.
 \end{align*}
 
 We carry out the Nass test (\cite{nass1959chi}) used in \cite{kratz2018multinomial} for its relatively better performance: 
 Let $S_m = \sum_{j=0}^m \frac{(O_{j+1} - N(\beta_{j+1}-\beta_j))^2}{N(\beta_{j+1} - \beta_j)}$, $c = \frac{\EE[S_m]}{\mathrm{var}[S_m]}$, $\nu = c\EE[S_m]$, where $\EE[S_m] = m$ and $\mathrm{var}[S_m] = 2m - \frac{m^2+4m+1}{N} + \frac{1}{N}\sum_{j=0}^m \frac{1}{\beta_{j+1}-\beta_j}.$ Then $cS_m$ follows $\chi^2_{\nu}$ under $H_0: \theta_j=\beta_j$ for $j=1,...,m$.

\section{A real application} \label{sec:application}

This section compares the performances of VCoVaR, VCoES, and the related contribution measures based on the historical daily closing price in USD spanned from 09/01/2015 to 02/06/2024 of the CC data.\footnote{The source of the price data is CoinMetrics.} We use the daily log-loss data (negative log-return). The sample size is $3082$ covering five CCs: BTC, ETH, LTC, XMR, and XRP, with the aggregate market capitalization around $68\%$ (02/07/2024).  

We follow the similar estimation procedures with the empirical study in \cite{waltz2022vulnerability} in selecting the univariate models and the copula models. Given the presence of a time-varying dependence structure in the distressed scenario involving the five CCs, we employ the dynamic DCC-copula model as utilized in \cite{waltz2022vulnerability}.  We fix $\alpha = \beta = 0.95$ and estimate the VCoVaR, VCoES, MCoVaR, MCoES of one CC with the conditional variables  modelling the other four CCs.

We have the following findings:
\begin{enumerate}[(i)]
    \item As shown in Figures \ref{LTC-fig} and \ref{BTC-fig}, both VCoVaR and VCoES exhibit similar patterns over time, with notable peaks during key market events.  In general, VCoES is a  more conservative systemic risk measure, reacting more sharply to market stress. This increased sensitivity is particularly evident during periods of heightened volatility, where VCoES displays a sharper rise and fall, indicating a greater responsiveness to extreme scenarios than VCoVaR. This is especially evident as indicated in  the BTC plot where the highest systemic risk is observed in 2020 during the COVID-19 outbreak.
    \item Compared to VCoVaR and VCoES in Figure \ref{xmr-fig}, $\Delta \vcovar$ and $\Delta \vcoes$ in Figure \ref{delta_xmr} are more effective in assessing the spillover effect under stressed conditions. The spillover effect is most severe during the Covid event in 2020. We observe that $\Delta\vcoes$ is more sensitive to extreme market scenarios, as it generally (though not always) dominates $\Delta\vcovar$.
    \item Unlike the other four decentralized cryptocurrencies (CCs), XRP is significantly influenced by Ripple Labs, the company that created and manages it. Figure \ref{XRP-ratio-fig} demonstrates that, generally (though not always), $\Delta^R\vcovar$ tends to exceed $\Delta^R\vcoes$. In comparison to the other four CCs, Figure \ref{XRP-delta-fig} shows that XRP exhibits the highest maximum values for both $\Delta\vcovar$ and $\Delta\vcoes$ over the period from 2016 to 2024, reaching approximately 0.7. One of the most significant events for XRP occurred in 2017, when its price surged from around \$0.006 at the beginning of the year to an all-time high of approximately \$3.84 by January 2018. Although the XRP Ledger is technically decentralized, Ripple Labs retains control over a substantial portion of the total XRP supply and plays a pivotal role in its development. This centralization introduces heightened risk, particularly when other cryptocurrencies are under distress. In such scenarios, liquidity across the market could be constrained. Given Ripple Labs' large XRP holdings, any action taken by the company or major stakeholders to manage their positions in response to market stress could trigger significant price volatility.
    \item The measure $\Delta^R\vcovar$ is useful for comparing the systemic risk positions of different assets. Figure \ref{BTC-fig} shows that the $\Delta^R\vcovar$ of BTC consistently dominates the other four cryptocurrencies (CCs), reflecting its market dominance, except at the beginning of 2018. During this period, the $\Delta^R\vcovar$ of BTC briefly falls below that of the other cryptocurrencies, coinciding with the aftermath of 2017 when Bitcoin’s price peaked and then dropped sharply. At this time, the other cryptocurrencies may have experienced relatively greater risk increments as they reacted to Bitcoin’s volatility and broader market conditions. The overlapping plots of the other four CCs from 2016 to 2018 are due to the market’s relative immaturity, which led to higher correlation and similar risk levels among these assets when the market was under distress. The period from late 2017 to early 2018 is widely regarded as a turning point for the overall cryptocurrency market. This phase was marked by the collapse of Bitcoin’s price, the bursting of the initial coin offering (ICO) bubble, and heightened regulatory scrutiny (\cite{wikipedia_crypto_bubble} and \cite{wired_future_bitcoin}).
 
 \end{enumerate}

\begin{figure}[htbp]
    \centering
    \includegraphics[width=\textwidth]{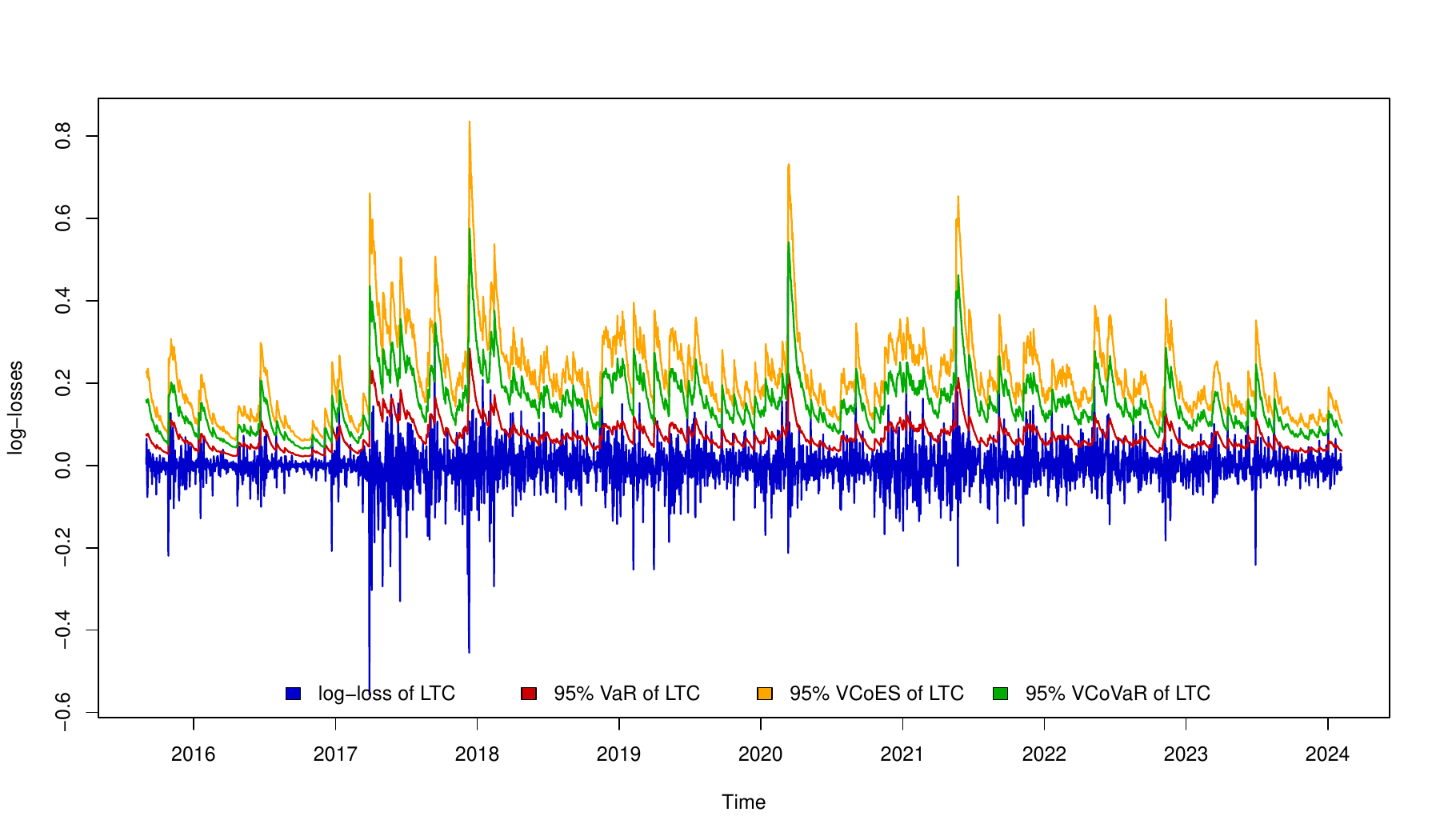}\caption{Comparisons of VaR, VCoVaR, and VCoES of LTC   under stressed condition using the DCC-copula}   
    \label{LTC-fig}
 \end{figure}

  \begin{figure}[htbp]
    \centering
    \includegraphics[width=\textwidth]{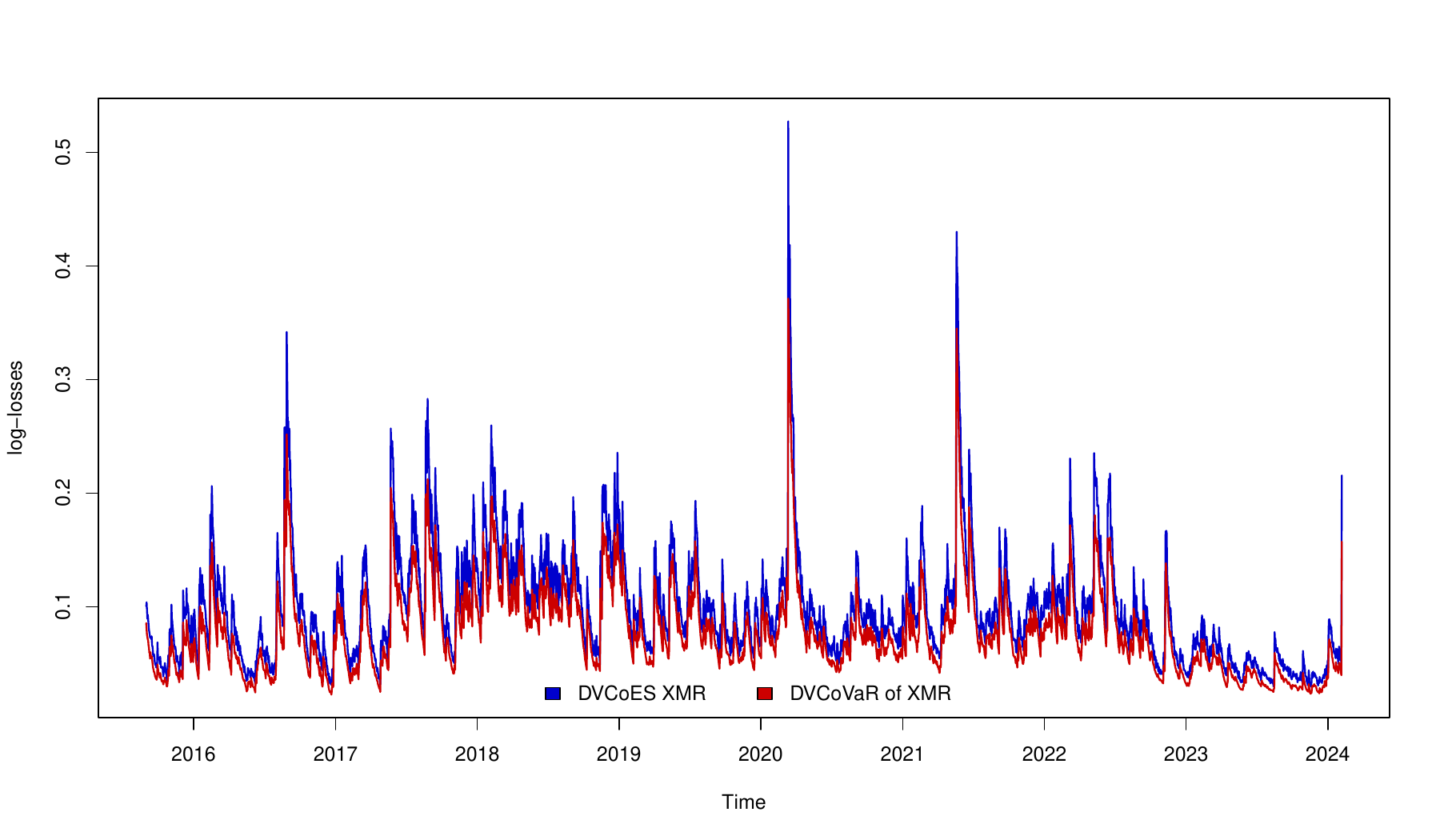}\caption{Comparison of $\Delta\vcovar$ and $\Delta\vcoes$ of  XMR under stressed condition using the DCC-copula} 
    \label{delta_xmr}\end{figure}

   \begin{figure}[htbp]
    \centering
    \includegraphics[width=\textwidth]{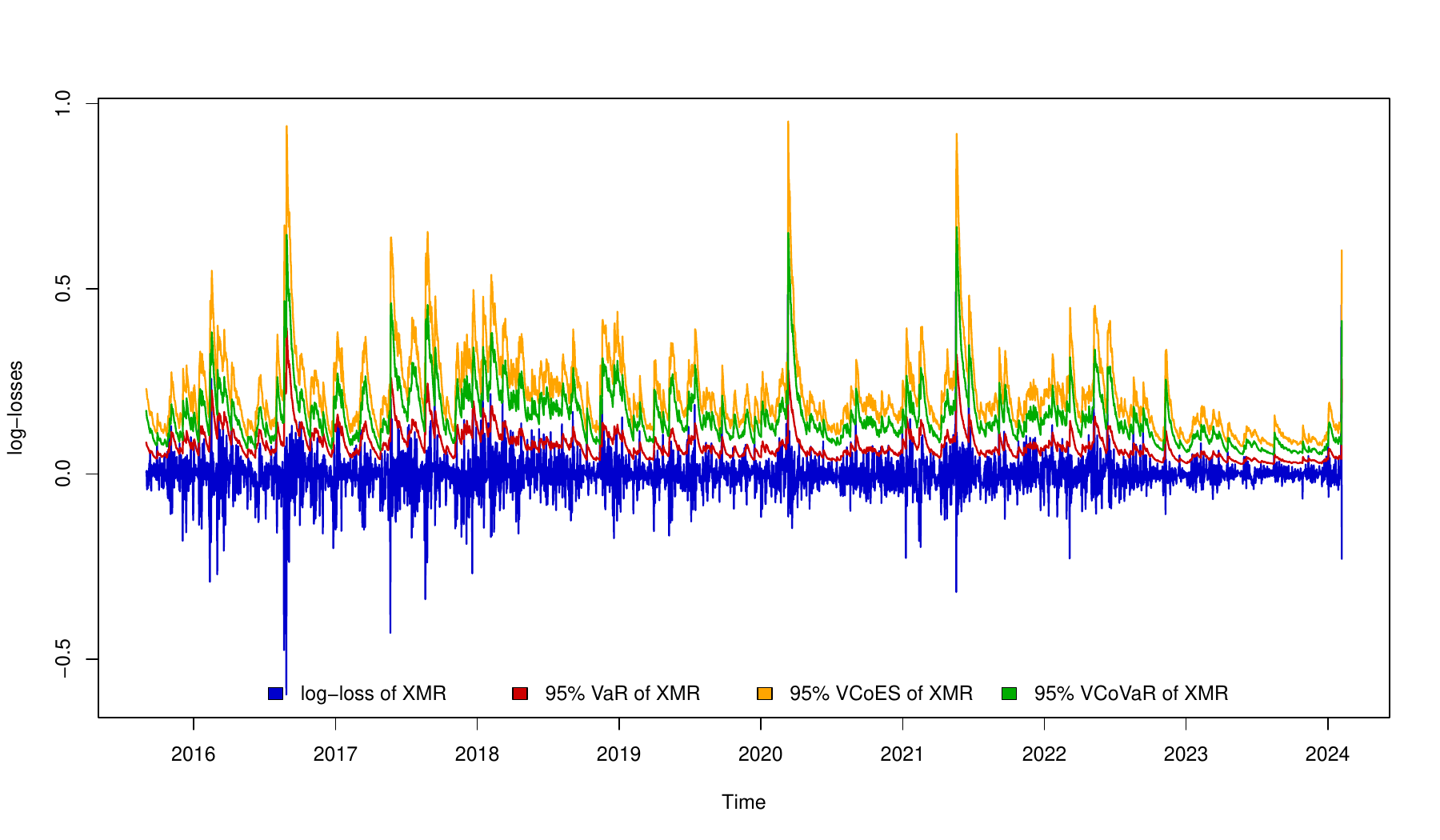}\caption{Comparisons of VaR, VCoVaR, and VCoES of XMR  under stressed condition using the DCC-copula}  \label{xmr-fig}  \end{figure}

   \begin{figure}[htbp]
    \centering
    \includegraphics[width=\textwidth]{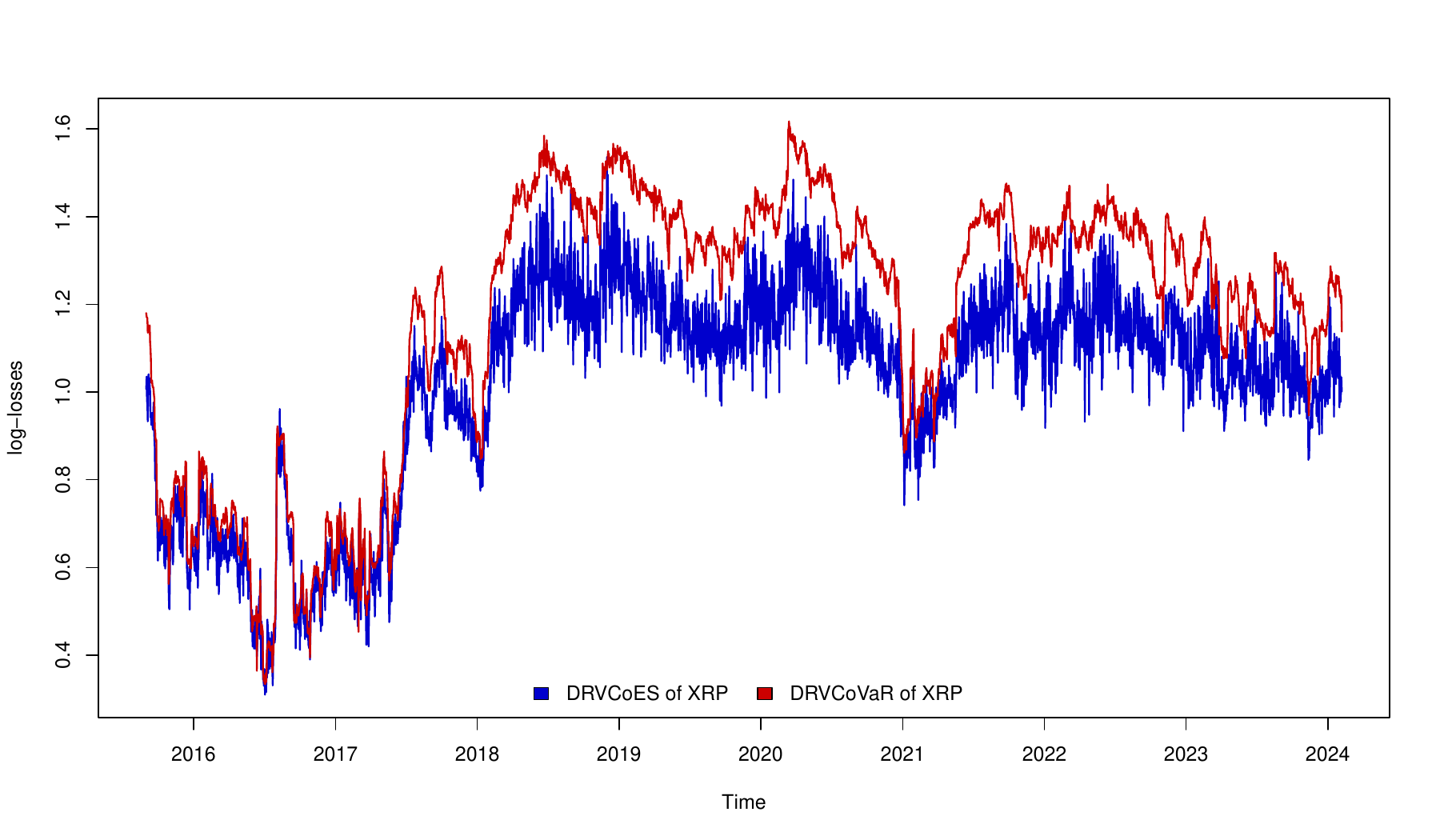}\caption{Comparisons of $\Delta^R\vcovar$ and $\Delta^R\vcoes$ of  XRP,  under stressed condition using the DCC-copula}   
        \label{XMR-ratio-fig}
 \end{figure}

 \begin{figure}[htbp]
    \centering
    \includegraphics[width=\textwidth]{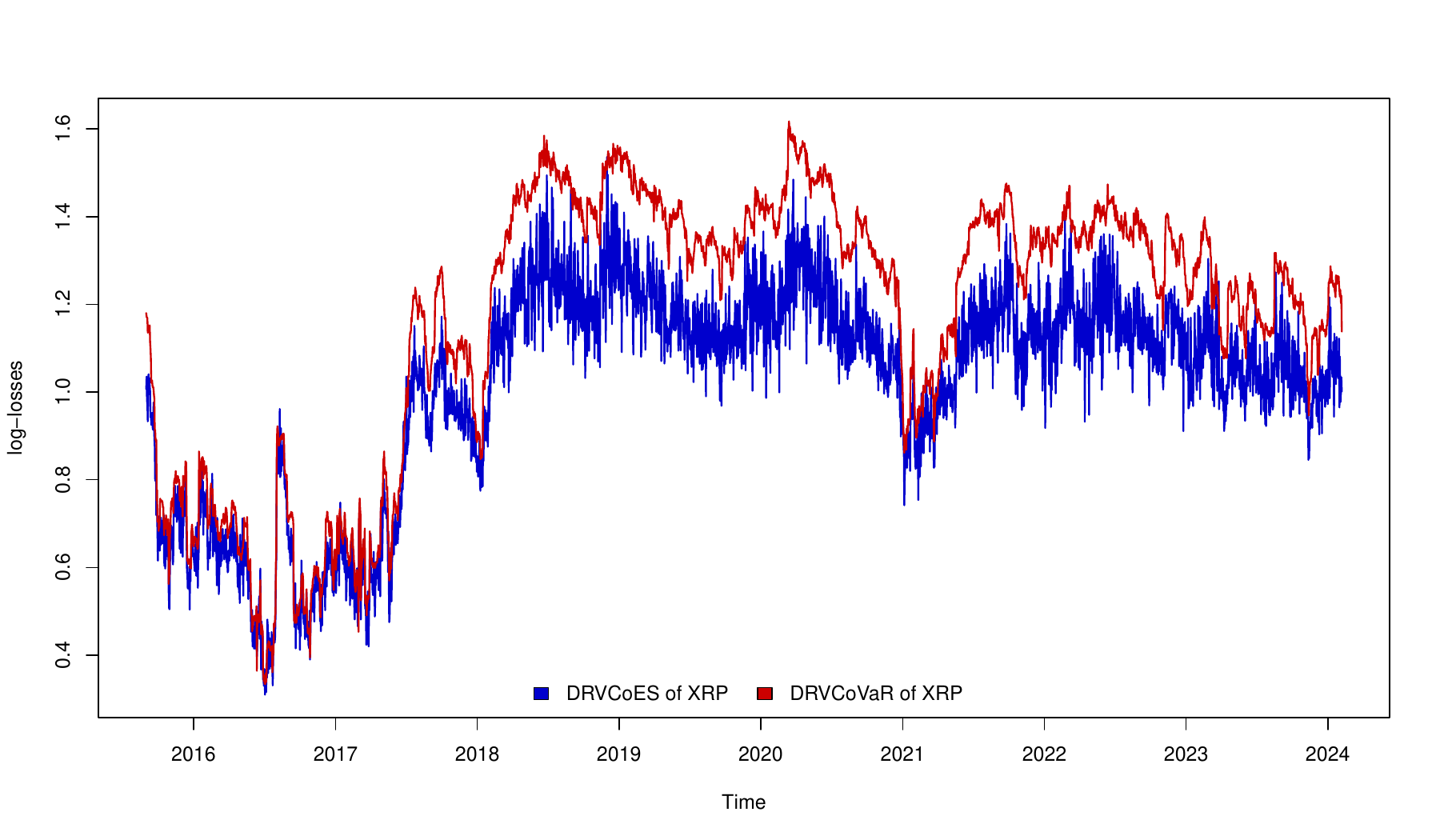}\caption{Comparisons of $\Delta^R\vcovar$ and $\Delta^R\vcoes$ of  XRP,  under stressed condition using the DCC-copula}   
        \label{XRP-ratio-fig}
 \end{figure}

  \begin{figure}[htbp]
    \centering
    \includegraphics[width=\textwidth]{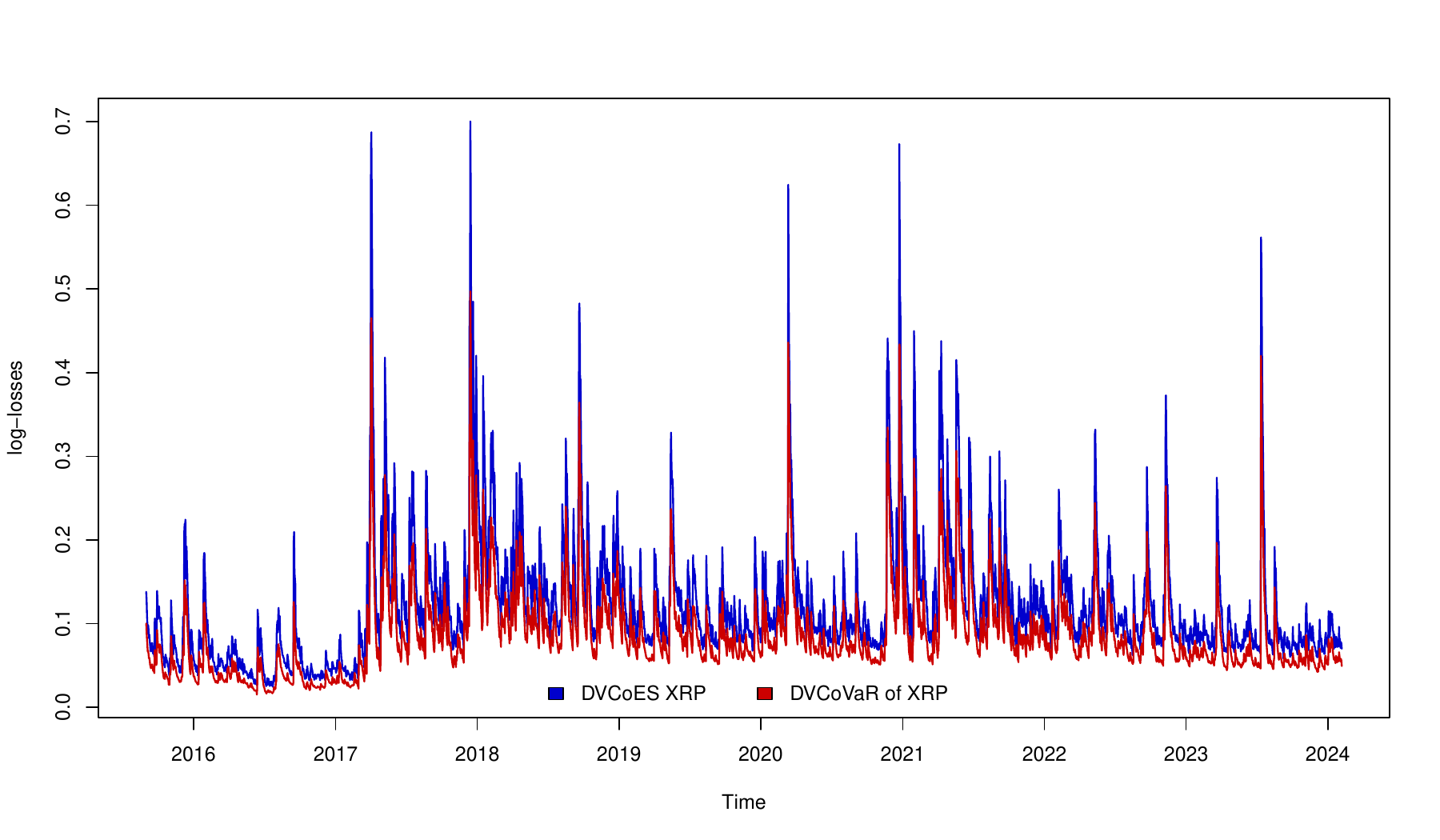}\caption{Comparisons of $\Delta\vcovar$ and $\Delta\vcoes$ of  XRP,  under stressed condition using the DCC-copula}   \label{XRP-delta-fig}
 \end{figure}

     \begin{figure}[htbp]
    \centering
    \includegraphics[width=\textwidth]{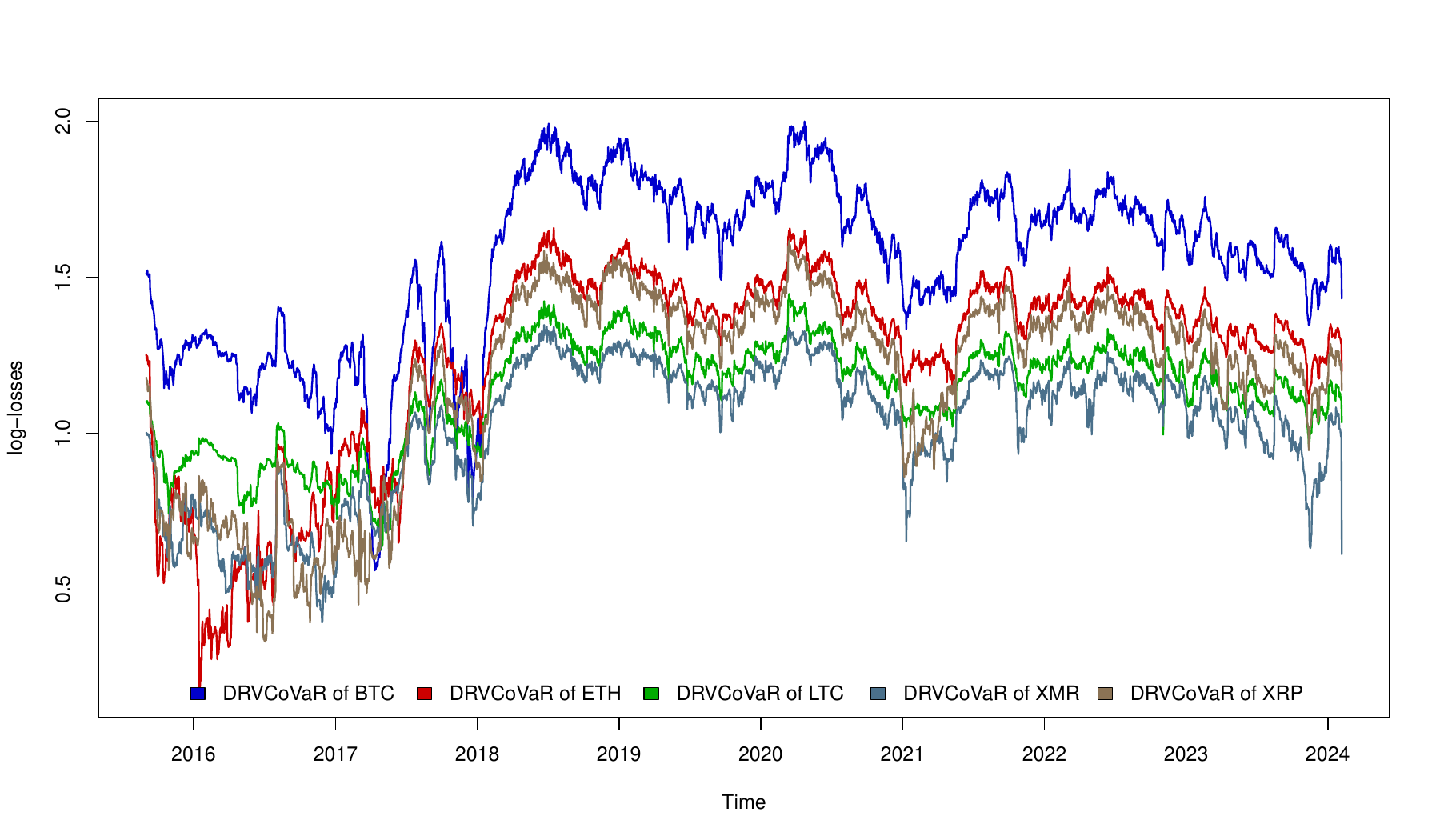}\caption{Comparisons of $\Delta^R\vcovar$ of  BTC, ETH, LTC, XMR, and XRP under stressed condition using the DCC-copula}   
        \label{BTC-fig}
 \end{figure}

To test the validity of our estimates, we carry out the backtesting procedures as described in section \ref{sec:backtesting} and choose $m=4$ to conduct the multinomial backtest of VCoES. The table below contains the violation rates of VCoVaR of each CC under the distressed scenario. All rates are close to $0.05$.
\begin{table}[h!]
    \centering
    \caption{Violation rates of VCoVaR under DCC copula}
    \begin{tabular}{c|c|c|c|c}
        \hline
        BTC & ETH & LTC & XMR & XRP \\ \hline
        0.0374 & 0.0602 & 0.0712 & 0.0557 & 0.0593 \\ \hline
    \end{tabular}
\end{table}
\\
The multinomial backtest result of VCoES of each CC under the distressed scenario is presented in the table below. The columns $O_0,...,O_m$ contain the observed numbers and $p_M$ gives the p-values of a multinomial Nass test. The p-values indicate our model is acceptable. 
\begin{table}[h!]
    \centering
    \caption{Multinomial backtest results of VCoES under DCC copula}
    \begin{tabular}{c|c|c|c|c|c|c|c}
        \hline
        & $N$ & $O_o$ & $O_1$ & $O_2$ & $O_3$ & $O_4$ & $p_M$ \\ \hline
       BTC & 348 & 335 & 7&3&3 & 0 & 0.5497 \\
       ETH & 349 & 327 & 11 & 8 & 3 & 0 & 0.1467 \\
       LTC & 365 & 339 & 6 & 12 & 8 & 0 & 0.1659 \\
       XMR & 341 & 322 & 5 & 6 & 8 & 0 & 0.8928 \\
       XRP & 354 & 332 & 6 & 10 & 6 & 0 & 0.5667 \\ \hline
    \end{tabular}
\end{table}


\section*{Acknowledgements}
The authors thank Yang Lu (Concordia)   for helpful discussions and valuable comments. Yunran Wei thanks the Department of Mathematics at SUSTech for its hospitality during her visit in August 2023. Yunran Wei  acknowledges financial support from the Natural Sciences and Engineering Research
Council of Canada (RGPIN-2023-04674, DGECR-2023-00454), the Research Development Grants, and the start-up fund at Carleton University.



\appendix
\section{Supplementary Definitions and Results}
\subsection{MCoVaR and MCoES}
Similar to the copula function defined in \eqref{eq:copula-def}, there exists a unique $d$-dimensional copula $\hat{C}$ such that for all
$(x_1, x_2, \cdots, x_d) \in \mathbb{R}^d$ it holds that 
\begin{equation} \nonumber
    \overline{F}\left(x_1, \ldots, x_d\right)=\hat{C}\left(1-F_1\left(x_1\right), \ldots, 1-F_d\left(x_d\right)\right), \forall (x_1, x_2, \cdots, x_d) \in \mathbb{R}^d.
\end{equation}
the function $\hat{C}$ is called the survival copula of a random vector $(X_1, X_2, \cdots, X_d)$(cf. \cite{Nelsen2007}, \cite{mai2017simulating}). Knowing the copula of a random vector allows us to compute its survival copula. For $u_1, u_2, \cdots, u_d \in (0,1)$, it follows that
\begin{equation} \nonumber
    \hat{C}(u_1, \cdots, u_d)=1+\sum_{k=1}^d(-1)^k \sum_{1 \leq j_1<\ldots<j_k \leq d} C_{j_1, \ldots, j_k}\left(1-u_{j_1}, \ldots, 1-u_{j_k}\right).
\end{equation}
The multivariate conditional value-at-risk (written as MCoVaR) of risk $Y$ given the stress event that all of $\bm{X}=(X_1,\ldots,X_d)$ exceed their stress levels is defined as follows:
\begin{definition}
Given the stress events that $\bm{X}=(X_1,\ldots,X_d)$ exceed their stress levels regulated by VaR's with confidence levels $\alpha_1,\ldots,\alpha_d$, the multivariate conditional value-at-risk (in short MCoVaR): $\mathcal{X} \rightarrow \R$ is given by
\begin{equation}\label{MCoVaR-defi}
{\rm MCoVaR}_{\bm{\alpha},\beta}(Y|\bm{X})={\rm VaR}_{\beta}(Y|\mbox{$\forall$ $i$: $X_i>{\rm VaR}_{\alpha_i}(X_i)$}),
\end{equation}
where $\alpha_i\in [0,1)$ and $\beta\in(0,1)$.
\end{definition}

Define 
$$A^{>}_{\bm{X}}=\{\mbox{$\forall$ $i$: $X_i>\mbox{VaR}_{\alpha_i}(X_i)$}\}\quad \mbox{and}\quad A^{>}_{\bm{U}}=\{\mbox{$\forall$ $i$: $U_i>\alpha_i$}\}.$$
Then the distribution function of $V|A^>_{\bm{U}}$ is given by 
\begin{equation} \label{eq:F-V-AUg}
    F_{V|A^>_{\bm{U}}}(v)=1-\frac{\hat{C}(1-\alpha_1,1-\alpha_2, \cdots,1-\alpha_d,1-v)}{\hat{C}(1-\alpha_1,1-\alpha_2, \cdots,1-\alpha_d,1)},\quad v\in[0,1].
\end{equation}
Further, (\ref{eq:F-V-AUg}) is also a  distortion function. The distribution function of $Y|A^>_{\bm{X}}$ is 
\begin{equation} \label{eq:F-Y-AXg}
    F_{Y|A^>_{\bm{X}}}(y)=F_{V|A^>_{\bm{U}}}(F_{Y}(y))=1-\hat{F}_{V|A^>_{\bm{U}}}(\overline{F}_{Y}(y)),\quad v\in[0,1],
\end{equation}
where $\hat{F}_{V|A^>_{\bm{U}}}$ is the dual of ${F}_{V|A^>_{\bm{U}}}$. Moreover, 
\begin{equation} \label{eq:mcovar-representation}
        {\rm MCoVaR}_{\bm{\alpha},\beta}(Y|\bm{X})=F^{-1}_{Y}\left(F^{-1}_{V|A^>_{\bm{U}}}(\beta)\right).
    \end{equation}
\begin{definition}\label{def:MCOES} The multivariate conditional expected shortfall (written as MCoES) of risk $Y$ given some stress event induced by $\bm{X}=(X_1,\ldots,X_d)$ is defined as follows:
    \begin{equation}\label{MCoES-defi}
        {\rm MCoES}_{\bm{\alpha},\beta}(Y|\bm{X})=\frac{1}{1-\beta}\int_{
        \beta}^1{\rm MCoVaR}_{\bm{\alpha},t}(Y|\bm{X})\dif t.
    \end{equation}
    \end{definition}
    It can be derived that 
    \begin{equation} \label{def:MCOES2}
        {\rm MCoES}_{\bm{\alpha},\beta}(Y|\bm{X}) = \int_{0}^1F^{-1}_{Y}(p)\dif \bar{h}_{TVaR}(F_{V|A^>_{\bm{U}}}(p))
    \end{equation}
    where  $F_{V|A^>_{\bm{U}}}(\cdot)$ is  defined in (\ref{eq:F-V-AUg}) and $\bar{h}_{TVaR}$ is given by 
\begin{equation}\label{eq:tvar-distortion}
\bar{h}_{TVaR}(t)=1-h_{TVaR}(1-t)=1-\min\{1,\frac{1-t}{1-\beta}\}=\max\{0,\frac{t-\beta}{1-\beta}\}.
\end{equation}
\subsection{Contribution Measures Using Different Baselines }

\begin{definition}\label{def:vulnercontri2} The vulnerability contribution conditional value-at-risk of risk $Y$ given the event $A_{\bm{X}}$ compared with the conditional value-at-risk of $Y$ given $X_{i}>{\rm VaR}_{\alpha_i}(X_i)$, for some $i\in\{1,2,\ldots,d\}$, written as $\Delta_i{\rm VCoVaR}$, is defined as follows:
    \begin{equation}\label{contriVCoVaR-defi2}
    \Delta_i{\rm VCoVaR}_{\bm{\alpha},\beta}(Y|\bm{X})={\rm VCoVaR}_{\bm{\alpha},\beta}(Y|\bm{X})-{\rm CoVaR}_{\alpha_i,\beta}(Y|X_i).
    \end{equation}
    Moreover, the associated vulnerability contribution ratio conditional value-at-risk is defined as follows:
    \begin{equation}\label{contriratioVCoVaR-defi2}
    \Delta^R_i{\rm VCoVaR}_{\bm{\alpha},\beta}(Y|\bm{X})=\frac{{\rm VCoVaR}_{\bm{\alpha},\beta}(Y|\bm{X})-{\rm CoVaR}_{\alpha_i,\beta}(Y|X_i)}{{\rm CoVaR}_{\alpha_i,\beta}(Y|X_i)}.
    \end{equation}
    \end{definition}

    According to \eqref{eq:vcovar-representation}, alternative expressions of (\ref{contriVCoVaR-defi2}) and (\ref{contriratioVCoVaR-defi2}) are given as follows:
    \begin{equation}\label{DiConVCoVaR-expre2}
    \Delta_i{\rm VCoVaR}_{\bm{\alpha},\beta}(Y|\bm{X})=F^{-1}_{Y}\left(F^{-1}_{V|A_{\bm{U}}}(\beta)\right)-F^{-1}_{Y}\left(F^{-1}_{V|U_i>\alpha_i}(\beta)\right)
    \end{equation}
    and 
    \begin{equation}\label{DiConratioVCoVaR-expre2}
    \Delta_i^R{\rm VCoVaR}_{\bm{\alpha},\beta}(Y|\bm{X})=\frac{F^{-1}_{Y}\left(F^{-1}_{V|A_{\bm{U}}}(\beta)\right)-F^{-1}_{Y}\left(F^{-1}_{V|U_i>\alpha_i}(\beta)\right)}{F^{-1}_{Y}\left(F^{-1}_{V|U_i>\alpha_i}(\beta)\right)}.
    \end{equation}
    
\par For given copula $C$, $\bm{\alpha} \in \mathbb{R}^d$ and $i \in \{1,2, \cdots, d\}$, let $\bm{\alpha}^*_i = (1, \cdots, 1, \alpha_i, 1, \cdots, 1)$ and
\begin{equation} \nonumber
    s^i_{\bm{\alpha}}(v)=\frac{v-C(\bm{\alpha},v)}{v-C(\bm{\alpha}^*_i,v)}.
\end{equation}
Note that the c.d.f. for conditional random variable $V \vert U_i > \alpha_i$ can be expressed as 
\begin{equation} \nonumber
    F_{V \vert U_i > \alpha_i}(v)= \frac{v-C(\bm{\alpha}^*_i,v)}{1-\alpha_i}.
\end{equation}
\begin{theorem} \label{thm:drivcovar-scdm}
    Suppose $C_1 = C_2 = C$ and $s^i_{\bm{\alpha}}(v) \leq s^i_{\bm{\alpha}}(1)$ for any $v \in [0,1]$. 
    \begin{enumerate}[(i)]
        \item If $Y_1 \leq_{{\rm disp}} Y_2$, then $\Delta_i{\rm VCoVaR}_{\bm{\alpha},\beta}(Y_1|\bm{X}_1) \leq \Delta_i{\rm VCoVaR}_{\bm{\alpha},\beta}(Y_2|\bm{X}_2)$ for all $\beta \in (0,1)$.
        \item If $Y_1 \leq_{{\star}} Y_2$, then $\Delta_i^{\rm R}{\rm VCoVaR}_{\bm{\alpha},\beta}(Y_1|\bm{X}_1) \leq \Delta_i^{\rm R}{\rm VCoVaR}_{\bm{\alpha},\beta}(Y_2|\bm{X}_2)$ for all $\beta \in (0,1)$.
    \end{enumerate}
\end{theorem}
\begin{theorem}\label{thm:drivcoes-scdm}
    Suppose $C_1 = C_2 = C$. 
    \begin{enumerate}[(i)]
        \item If $Y_1 \leq_{{\rm disp}} Y_2$ and $s^i_{\bm{\alpha}}(v) \leq s^i_{\bm{\alpha}}(1)$ for any $v \in [0,1]$, then $\Delta_i{\rm VCoES}_{\bm{\alpha},\beta}(Y_1|\bm{X}_1) \leq \Delta_i{\rm VCoES}_{\bm{\alpha},\beta}(Y_2|\bm{X}_2)$ for all $\beta \in (0,1)$.
        \item If $Y_1 \leq_{{\rm eps}} Y_2$, $C$ is concave w.r.t. its last argument and $F_{V|A_{\bm{U}}} \circ F^{-1}_{V|U_i>\alpha_i}$ is convex, then $\Delta_i^{\rm R}{\rm VCoES}_{\bm{\alpha},\beta}(Y_1|\bm{X}_1) \leq \Delta_i^{\rm R}{\rm VCoES}_{\bm{\alpha},\beta}(Y_2|\bm{X}_2)$ for all $\beta \in (0,1)$.
    \end{enumerate}
\end{theorem}

\section{Proofs of the main results}
Here we provide necessary lemmas that are utilized in this section.
\begin{lemma} \label{lem:cdf-ineq}
    $V^*|A_{\bm{U}^*} \leq_{\rm st} \hat{V}|A_{\hat{\bm{U}}}$ iff 
    \begin{equation} \nonumber
        l_{\bm{\alpha}}(v) \geq l_{\bm{\alpha}}(1), \forall v \in [0,1]. 
    \end{equation}
\end{lemma}
\proof It can be derived from $l_{\bm{\alpha}}(v) \geq l_{\bm{\alpha}}(1)$ for all $v \in [0,1]$ that 
\begin{equation} \nonumber
    \frac{v-C_1(\bm{\alpha},v)}{1-C_1(\bm{\alpha},1)} \geq \frac{v-C_2(\bm{\alpha},v)}{1-C_2(\bm{\alpha},1)}, ~ \forall v \in [0,1],
\end{equation}
which implies $F_{V^*|A_{\bm{U}^*}}(v) \geq F_{\hat{V}|A_{\hat{\bm{U}}}}(v)$ for all $v \in [0,1]$.\qed
\begin{lemma} \label{lem:LTD-ineq}
    Suppose a $d$-demensional copula $C$ is ${\rm LTD}_{d+1}^1$, then we $V \leq_{\rm st} V|A_{\bm{U}}$.
\end{lemma}
\proof Let $U_1,\ldots,U_d,V$ be a group of uniform random variables such that $(\bm{U},V)$ also shares copula $C$, where $\bm{U}=(U_1,\ldots,U_d)$. Definition \ref{def:LTD} gives rise to 
    \begin{equation} \nonumber
        \PP\{ U_1 \leq \alpha_1, \cdots, U_d \leq \alpha_d \vert V \leq v\} \geq \PP\{ U_1 \leq \alpha_1, \cdots, U_d \leq \alpha_d\},
    \end{equation}
    which implies 
    \begin{equation} \nonumber
         \frac{C(\alpha_1,\ldots,\alpha_d,v)}{v} \geq C(\alpha_1,\ldots,\alpha_d,1)
    \end{equation}
    and further the desired result. \qed
\subsection{Proof of Lemma \ref{lem:cdf-representation}}
\proof For $V \vert A_{\bm{U}}$, it can be derived that 
    \begin{equation} \nonumber
    \mathbb{P}(V \leq v\mid A_{\bm{U}}) = \frac{\mathbb{P}(V \leq v , A_{\bm{U}})}{\mathbb{P}(A_{\bm{U}})} = \frac{\mathbb{P}(V \leq v)-\mathbb{P}(V \leq v, A^c_{\bm{U}})}{1-\mathbb{P}(A^c_{\bm{U}})}.
    \end{equation}
    Considering 
    \begin{equation} \nonumber
        \mathbb{P}(V \leq v, A^c_{\bm{U}}) = C(\alpha_1,\cdots, \alpha_d,v)
    \end{equation}
    and 
    \begin{equation} \nonumber
        \mathbb{P}(A^c_{\bm{U}}) = C(\alpha_1,\cdots, \alpha_d,1),
    \end{equation}
    the result of (\ref{eq:F-V-AU}) can be obtained. 
    
    \par The proof for (b) is trivial and hence omitted.
     \qed
\subsection{Proof of Lemma \ref{thm:vcovar-dcdm}}
\proof Based on \eqref{eq:vcovar-representation}, it suffices to show 
\begin{equation} \nonumber
    F^{-1}_{Y_1}\left(F^{-1}_{V^*|A_{\bm{U}^*}}(\beta)\right) \leq F^{-1}_{Y_2}\left(F^{-1}_{\hat{V}|A_{\hat{\bm{U}}}}(\beta)\right)
\end{equation}
for all $\beta \in (0,1)$. Considering $Y_1 \leq_{{\rm st}} Y_2$, it is obvious that $F^{-1}_{Y_1}\left(t\right) \leq F^{-1}_{Y_2}\left(t\right)$ for all $t \in (0,1)$. Considering $l_{\bm{\alpha}}(v) \geq l_{\bm{\alpha}}(1)$ for all $v \in [0,1]$ and applying Lemma \ref{lem:cdf-ineq} give rise to $F_{V^*|A_{\bm{U}^*}}(v) \geq F_{\hat{V}|A_{\hat{\bm{U}}}}(v)$. Hence, it can be derived that
\begin{equation} \label{eq:F-1geq}
    F^{-1}_{V^*|A_{\bm{U}^*}}(\beta) \leq F^{-1}_{\hat{V}|A_{\hat{\bm{U}}}}(\beta), \quad \forall \beta \in (0,1).
\end{equation}
Further, $Y_1 \leq_{{\rm st}} Y_2$ implies $F^{-1}_{Y_1}\left(t\right) \leq F^{-1}_{Y_2}\left(t\right)$ for all $t \in (0,1)$ and then the desired result. \qed.
\subsection{Proof of Corollary \ref{cor:VCoVaR-greater}}
\proof For given $\bm{\alpha} \in [0,1]^d$, the ${\rm LTD}_{d+1}^1$ property of $C_1$ give rise to 
    \begin{equation} \label{eq:ind}
        \frac{C_1(\bm{\alpha},v)}{v} \geq C_1(\bm{\alpha},1).
    \end{equation}
    Based on \eqref{eq:ind} it can be derived that 
    \begin{equation} \label{eq:l-alpha-ind}
        \frac{v(1 - \Pi_{i=1}^d \alpha_i)}{v - C_1(\bm{\alpha},v)} \geq \frac{1 - \Pi_{i=1}^d \alpha_i}{1 - C_1(\bm{\alpha},1)}.
    \end{equation}
    Then the desired result can be obtained by using \eqref{eq:l-alpha-ind}, noticing $Y_1 \leq_{\rm st} Y_1$ and applying Theorem \ref{thm:vcovar-dcdm}. \qed
\subsection{Proof of Lemma \ref{thm:d-drvcovar-dcdm}}
\proof Without loss of generality, suppose $C_2$ is ${\rm LTD}_{d+1}^1$. Based on Lemma \ref{lem:LTD-ineq}, we have $F^{-1}_{\hat{V}|A_{\hat{\bm{U}}}}(\beta) \geq \beta$ for all $\beta \in (0,1)$.
\par \underline{Proof of (i)}: Based on \eqref{ConratioVCoVaR-expre}, it suffices to show
        \begin{equation} \label{eq:drvcovar-compare-repre}
            F^{-1}_{Y_1}(F^{-1}_{{V}^*|A_{{\bm{U}^*}}}(\beta))-F^{-1}_{Y_1}(\beta) \leq F^{-1}_{Y_2}(F^{-1}_{\hat{V}|A_{\hat{\bm{U}}}}(\beta))-F^{-1}_{Y_2}(\beta), ~ \forall \beta \in (0,1).
        \end{equation}
        Inequality $l_{\bm{\alpha}}(v) \geq l_{\bm{\alpha}}(1)$ for any $v \in [0,1]$ implies \eqref{eq:F-1geq}, which further give rise to 
        \begin{equation} \label{eq:drvcovar-compare-1}
            F^{-1}_{Y_1}(F^{-1}_{{V}^*|A_{{\bm{U}^*}}}(\beta))-F^{-1}_{Y_1}(\beta) \leq F^{-1}_{Y_1}(F^{-1}_{\hat{V}|A_{\hat{\bm{U}}}}(\beta))-F^{-1}_{Y_1}(\beta), ~ \forall \beta \in (0,1).
        \end{equation}
        Noticing $F^{-1}_{\hat{V}|A_{\hat{\bm{U}}}}(\beta) \geq \beta$,  
        \begin{equation}\label{eq:drvcovar-compare-2}
                F^{-1}_{Y_1}(F^{-1}_{\hat{V}|A_{\hat{\bm{U}}}}(\beta))-F^{-1}_{Y_1}(\beta) \leq F^{-1}_{Y_2}(F^{-1}_{\hat{V}|A_{\hat{\bm{U}}}}(\beta))-F^{-1}_{Y_2}(\beta), ~ \forall \beta \in (0,1).
        \end{equation}
        can be derived by considering Definition \ref{def:order} and $Y_1 \leq_{{\rm disp}} Y_2$. Combining \eqref{eq:drvcovar-compare-1} and \eqref{eq:drvcovar-compare-2} gives rise to \eqref{eq:drvcovar-compare-repre} and further the desired result.
   \par \underline{Proof of (ii)}: It can be established that $F^{-1}_{Y_2}(v)/F^{-1}_{Y_1}(v)$ is increasing w.r.t. $v \in (0,1)$ since $Y_1 \leq_{{\star}} Y_2$. Then one obtains
        \begin{equation} \label{eq:drvcovar-pf1}
            \frac{F^{-1}_{Y_2}(F^{-1}_{\hat{V}|A_{\hat{\bm{U}}}}(\beta))}{F^{-1}_{Y_1}(F^{-1}_{\hat{V}|A_{\hat{\bm{U}}}}(\beta))} \geq \frac{F^{-1}_{Y_2}(\beta)}{F^{-1}_{Y_1}(\beta)}, ~ \forall \beta \in (0,1)
        \end{equation}
        since $F^{-1}_{\hat{V}|A_{\hat{\bm{U}}}}(\beta) \geq \beta$. Inequality $l_{\bm{\alpha}}(v) \geq l_{\bm{\alpha}}(1)$ for any $v \in [0,1]$ implies \eqref{eq:F-1geq}, which further shows 
        \begin{equation}\label{eq:drvcovar-pf2}
            \frac{F^{-1}_{Y_2}(F^{-1}_{\hat{V}|A_{\hat{\bm{U}}}}(\beta))}{F^{-1}_{Y_1}(F^{-1}_{{V}^*|A_{{\bm{U}^*}}}(\beta))} \geq \frac{F^{-1}_{Y_2}(F^{-1}_{\hat{V}|A_{\hat{\bm{U}}}}(\beta))}{F^{-1}_{Y_1}(F^{-1}_{\hat{V}|A_{\hat{\bm{U}}}}(\beta))}, ~ \forall \beta \in (0,1).
        \end{equation}
        Inequalities \eqref{eq:drvcovar-pf1} and \eqref{eq:drvcovar-pf2} yield 
        \begin{equation}\nonumber
            \frac{F^{-1}_{Y_2}(F^{-1}_{\hat{V}|A_{\hat{\bm{U}}}}(\beta))}{F^{-1}_{Y_2}(\beta)} \geq \frac{F^{-1}_{Y_1}(F^{-1}_{{V}^*|A_{{\bm{U}^*}}}(\beta))}{F^{-1}_{Y_1}(\beta)} , ~ \forall \beta \in (0,1),
        \end{equation}
        which gives the desired result.\qed
\subsection{Proof of Lemma \ref{thm:dr-vcoes-dcdm}}

For a r.v. $X$ with c.d.f. $F$, we define
\begin{equation}\label{IABfunc}
  I_{A, B}(X)=\frac{\int_0^1 F^{-1}(t) \dif A(t)}{\int_0^1 F^{-1}(t) \dif B(t)}-1,
\end{equation}
where both $A(t)$ and $B(t)$ are distortion functions, that is, $A(t)\in\mathcal{H}$ and $B(t)\in\mathcal{H}$. Let $$\mathcal{C}_1=\{I_{A, B}:A(t)\in\mathcal{H},~B(t)\in\mathcal{H},~\mbox{$A\circ B^{-1}(t)$ is convex}\}$$
and
$$\mathcal{C}_2=\{I_{A, B}:A(t)\in\mathcal{H},~B(t)\in\mathcal{H},~\mbox{both $A\circ B^{-1}(t)$ and $B(t)$ are convex}\}.$$
Clearly, $\mathcal{C}_2$ is a subset of $\mathcal{C}_1$. The following lemma is due to Theorem 3.25 of \cite{belzunce2012comparison}, which establishes an equivalent characterization for the expected proportional shortfall order in terms of the class of well-defined ratio integrals (\ref{IABfunc}) within $\mathcal{C}_2$.
\par The following lemma provides a provides an equivalent characterization for the expected proportional shortfal order.
\begin{lemma}{\rm\citep[Theorem 3.25 of][]{belzunce2012comparison}}\label{lemma:epw} Let $X$ and $Y$ be two random variables with c.d.f.'s $F$ and $G$, respectively. Then, $X \leq_{\rm eps} Y$ if and only if $I_{A,B}(X)\leq I_{A,B}(Y)$ for all $I_{A,B}\in \mathcal{C}_2$.
\end{lemma}
\proof Without loss of generality, suppose $C_2$ is concave w.r.t. its last argument.
\par \underline{Proof of (i)}: the concavity of $F_{V|A_{\bm{U}}}(\cdot)$ defined in \eqref{eq:F-V-AU} can be established. The risk measure ${\rm VCoES}$ can be treated as a distortion risk measure with concave distortion function $\bar{h}_{TVaR}(F_{V|A_{\bm{U}}}(\cdot))$ (as showed in \eqref{VCoES-defi-Equi}). Combining Definition \ref{def:order} and $Y_1 \leq_{\rm icx} Y_2$, we have 
        \begin{equation} \label{eq:coes-compare-1}
            \int_{0}^1F^{-1}_{Y_1}(p)\dif \bar{h}_{TVaR}(F_{\hat{V}|A_{\hat{\bm{U}}}}(p)) \leq \int_{0}^1F^{-1}_{Y_2}(p)\dif \bar{h}_{TVaR}(F_{\hat{V}|A_{\hat{\bm{U}}}}(p)).
        \end{equation} 
        The proof of Theorem \ref{thm:vcovar-dcdm} shows that 
        \begin{equation} \nonumber
            F^{-1}_{V^*|A_{\bm{U}^*}}(\beta) \leq F^{-1}_{\hat{V}|A_{\hat{\bm{U}}}}(\beta), \quad \forall \beta \in (0,1).
        \end{equation}
        can be implied from $l_{\bm{\alpha}}(v) \geq l_{\bm{\alpha}}(1)$ for all $v \in [0,1]$. Then, 
        \begin{equation} \label{eq:coes-compare-2}
            \int_{0}^1F^{-1}_{Y_1}(F^{-1}_{V^*|A_{\bm{U}^*}}(p))\dif \bar{h}_{TVaR}(p) \leq \int_{0}^1F^{-1}_{Y_1}(F^{-1}_{\hat{V}|A_{\hat{\bm{U}}}}(p))\dif \bar{h}_{TVaR}(p).
        \end{equation}
        can be derived. The desired result can be derived by combining \eqref{VCoES-defi-Equi}, \eqref{eq:coes-compare-1} and \eqref{eq:coes-compare-2}.
        \par \underline{Proof of (ii)}: Inequality $l_{\bm{\alpha}}(v) \geq l_{\bm{\alpha}}(1)$ for any $v \in [0,1]$ shows $F_{V^*|A_{\bm{U}^*}}(v) \geq F_{\hat{V}|A_{\hat{\bm{U}}}}(v)$, which gives rise to
    \begin{equation}\label{eq:DRVCoES-compare-1}
        \frac{\int_{0}^1 F^{-1}_{Y_1}\left(p\right)\dif \bar{h}_{TVaR}(F_{V^*|A_{\bm{U}^*}}(p))}{\int_{\beta}^1 F^{-1}_{Y_1}(p) \dif \bar{h}_{TVaR}(p)} \leq \frac{\int_{0}^1 F^{-1}_{Y_1}\left(p\right)\dif \bar{h}_{TVaR}(F_{\hat{V}|A_{\hat{\bm{U}}}}(p))}{\int_{\beta}^1 F^{-1}_{Y_1}(p) \dif \bar{h}_{TVaR}(p)}.
    \end{equation}
    it has been known that both of $\bar{h}_{TVaR}(F_{V^*|A_{\bm{U}^*}}(p))$ and $\overline{h}_{TVaR}(p)$ are increasing and convex on $p\in[0,1]$, and the (generalized) inverse function of $\overline{h}_{TVaR}(t)$ is given by
    \begin{equation} \nonumber
        \overline{h}^{-1}_{TVaR}(t)=\inf \{x \in [0,1] \mid \overline{h}_{TVaR}(x) \geq t\}=\beta+(1-\beta)t,\quad t\in [0,1],
    \end{equation}
    which is linear. Therefore, one has $\bar{h}_{TVaR}(F_{V^*|A_{\bm{U}^*}}(\overline{h}^{-1}_{TVaR}(t)))$ is convex in $t\in[0,1]$ due to the PDS property of $C$. Then, based on Lemma \ref{lemma:epw}, we have
    \begin{equation}\label{eq:DRVCoES-compare-2}
        \frac{\int_{0}^1 F^{-1}_{Y_1}\left(p\right)\dif \bar{h}_{TVaR}(F_{\hat{V}|A_{\hat{\bm{U}}}}(p))}{\int_{\beta}^1 F^{-1}_{Y_1}(p) \dif \bar{h}_{TVaR}(p)} \leq \frac{\int_{0}^1 F^{-1}_{Y_2}\left(p\right)\dif \bar{h}_{TVaR}(F_{\hat{V}|A_{\hat{\bm{U}}}}(p))}{\int_{\beta}^1 F^{-1}_{Y_2}(p) \dif \bar{h}_{TVaR}(p)},
    \end{equation}
    which gives rise to the desired result.
\subsection{Proof of Lemma \ref{thm:dvcoes-dcdm}}
\proof Inequality $l_{\bm{\alpha}}(v) \geq l_{\bm{\alpha}}(1)$ for any $v \in [0,1]$ implies \eqref{eq:F-1geq}, which shows
    \begin{equation}\nonumber
        \int_{\beta}^1 F^{-1}_{Y_1}\left(F^{-1}_{V^*|A_{\bm{U}^*}}(t)\right) - F^{-1}_{Y_1}(t) \dif t \leq \int_{\beta}^1 F^{-1}_{Y_1}\left(F^{-1}_{\hat{V}|A_{\hat{\bm{U}}}}(t)\right) - F^{-1}_{Y_1}(t) \dif t,
    \end{equation}
One obtains $v \geq F_{\hat{V}|A_{\hat{\bm{U}}}}(v)$ from Lemma \ref{lem:LTD-ineq} and the condition that $C_2$ is ${\rm LTD}_{d+1}^1$, and hence $F^{-1}_{\hat{V}|A_{\hat{\bm{U}}}}(\beta) \geq \beta$. Based on $Y_1 \leq_{{\rm disp}} Y_2$, we have
    \begin{equation}\nonumber
        \int_{\beta}^1 F^{-1}_{Y_1}\left(F^{-1}_{\hat{V}|A_{\hat{\bm{U}}}}(t)\right) - F^{-1}_{Y_1}(t) \dif t \leq \int_{\beta}^1 F^{-1}_{Y_2}\left(F^{-1}_{\hat{V}|A_{\hat{\bm{U}}}}(t)\right) - F^{-1}_{Y_2}(t) \dif t,
    \end{equation}
which finishes the proof. \qed

\subsection{Proof of Theorem \ref{thm:drivcovar-scdm}}
It can be derived from $s^i_{\bm{\alpha}}(v) \leq s^i_{\bm{\alpha}}(1)$ for any $v \in [0,1]$ that 
\begin{equation} \label{eq:VUleqVUi}
    F_{V|A_{\bm{U}}}(v) = \frac{v-C(\bm{\alpha},v)}{1-C(\bm{\alpha},1)} \leq \frac{v-C(\bm{\alpha}^*_i,v)}{1-\alpha_i} = F_{V|U_i>\alpha_i}(v), \forall v \in [0,1].
\end{equation}
Further, \eqref{eq:VUleqVUi} gives rise to $F^{-1}_{V|A_{\bm{U}}}(\beta) \geq F^{-1}_{V|U_i>\alpha_i}(\beta)$ for all $\beta \in(0,1)$. 
\newline \underline{Proof of (i)}: Part (i) of Theorem \ref{thm:drivcovar-scdm} can be obtained by 
\begin{equation} \label{eq:drivcovar-scdm-pf1}
    F^{-1}_{Y_1}(F^{-1}_{V|A_{\bm{U}}}(\beta)) - F^{-1}_{Y_1}(F^{-1}_{V|U_i>\alpha_i}(\beta)) \leq F^{-1}_{Y_2}(F^{-1}_{V|A_{\bm{U}}}(\beta)) - F^{-1}_{Y_2}(F^{-1}_{V|U_i>\alpha_i}(\beta)),
\end{equation}
which is derived by considering (iii) of Definition \ref{def:order} and \eqref{DiConVCoVaR-expre2}. 
\newline \underline{Proof of (ii)}: It can be established that $F^{-1}_{Y_2}(v)/F^{-1}_{Y_1}(v)$ is increasing w.r.t. $v \in (0,1)$ since $Y_1 \leq_{{\star}} Y_2$. Then one obtains
\begin{equation} \nonumber
    \frac{F^{-1}_{Y_2}(F^{-1}_{V|A_{\bm{U}}}(\beta))}{F^{-1}_{Y_1}(F^{-1}_{V|A_{\bm{U}}}(\beta))} \geq \frac{F^{-1}_{Y_2}(F^{-1}_{V|U_i>\alpha_i}(\beta))}{F^{-1}_{Y_1}(F^{-1}_{V|U_i>\alpha_i}(\beta))}, ~ \forall \beta \in (0,1),
\end{equation}
which gives rise to 
\begin{equation} \nonumber
    \frac{F^{-1}_{Y_2}(F^{-1}_{V|A_{\bm{U}}}(\beta))}{F^{-1}_{Y_2}(F^{-1}_{V|U_i>\alpha_i}(\beta))} \geq \frac{F^{-1}_{Y_1}(F^{-1}_{V|A_{\bm{U}}}(\beta))}{F^{-1}_{Y_1}(F^{-1}_{V|U_i>\alpha_i}(\beta))}, ~ \forall \beta \in (0,1),
\end{equation}
and thus the desired result follows.
\qed
\subsection{Proof of Theorem \ref{thm:drivcoes-scdm}}
\underline{Proof of (i)}: The fact that $s^i_{\bm{\alpha}}(v) \leq s^i_{\bm{\alpha}}(1)$ for any $v \in [0,1]$ implies $F^{-1}_{V|A_{\bm{U}}}(\beta) \geq F^{-1}_{V|U_i>\alpha_i}(\beta)$ for all $\beta \in(0,1)$, as the proof of Theorem \ref{thm:drivcovar-scdm} points out. Then the desired result can be obtained by integrating  both sides of \eqref{eq:drivcovar-scdm-pf1} .
\newline \underline{Proof of (ii)}: Note that the function $\bar{h}_{TVaR}(F_{V|U_i>\alpha_i}(p))$ is increasing and convex on $p\in[0,1]$ since $C$ is concave w.r.t. the last argument. The convexity of function $\bar{h}_{TVaR} \circ F_{V|A_{\bm{U}}} \circ F^{-1}_{V|U_i>\alpha_i} \circ \bar{h}^{-1}_{TVaR}$ can be implied from the convexity of $F_{V|A_{\bm{U}}} \circ F^{-1}_{V|U_i>\alpha_i}$. Then, based on Lemma \ref{lemma:epw}, we have
\begin{equation}\nonumber
    \frac{\int_{0}^1 F^{-1}_{Y_1}\left(p\right)\dif \bar{h}_{TVaR}(F_{{V}|A_{{\bm{U}}}}(p))}{\int_{\beta}^1 F^{-1}_{Y_1}(p) \dif \bar{h}_{TVaR}(F_{V|U_i>\alpha_i}(p))} \leq \frac{\int_{0}^1 F^{-1}_{Y_2}\left(p\right)\dif \bar{h}_{TVaR}(F_{{V}|A_{{\bm{U}}}}(p))}{\int_{\beta}^1 F^{-1}_{Y_2}(p) \dif \bar{h}_{TVaR}(F_{V|U_i>\alpha_i}(p))},
\end{equation}
which gives rise to the desired result.
\qed

\bibliographystyle{mystyle}
{\small\bibliography{VCRM}}
\end{document}